\documentclass[11pt]{amsart}
\usepackage[latin1]{inputenc}
\usepackage[all]{xy}
\usepackage{amssymb, amsmath, amscd, geometry}
\usepackage{graphicx}
\usepackage{setspace}
\numberwithin{equation}{section}
\input xy
\xyoption{all}
\usepackage{latexsym}
\usepackage[usenames]{color}
\usepackage{epic} 
\usepackage{mathrsfs}
\usepackage{euscript}
\usepackage{rotating}

\renewcommand{\phi}{\varphi}
\newcommand{\lb}{\left (}
\newcommand{\rb}{\right )}

\newcommand{\suvp}{\sup_{U_{AA'},V_{BB'},\ket{\phi}}}
\newcommand{\suvc}{\sup_{U_{AA'},V_{BA'},\ket{\phi}}}
\newcommand{\suv}{\sup_{U_{AA'},V_{BB'}}}
\newcommand{\bra}[1]{\langle #1 \vert}
\newcommand{\ket}[1]{\vert #1 \rangle}

\newcommand{\N}[1]{\left \Vert #1 \right \Vert}

\newcommand{\st}{\, : \,}
\newcommand{\h}{\mathcal H}

\newcommand{\mc}{\mathcal}

\theoremstyle{plain}
\newtheorem{lemma}{Lemma}[section]
\newtheorem{theorem}[lemma]{Theorem}

\newtheorem{prop}[lemma]{Proposition}

\theoremstyle{definition}
\newtheorem{definition}[lemma]{Definition}
\newtheorem{remark}[lemma]{Remark}

      \newcommand{\C}{{\mathbb C}}

    \newcommand{\He}{{\mathcal H}}

\renewcommand{\>}{\rangle}
\newcommand{\<}{\langle}

\newcommand{\supp}{\operatorname{supp}}

\newcommand{\gl}{\, \ge \, }

\allowdisplaybreaks

\begin{document}

\baselineskip=17pt

\title[Rank-one Quantum Games]{Rank-one Quantum Games}


\author{T. Cooney}
\address{Departamento de An\'alisis Matem\'atico and IMI, Universidad Complutense de
Madrid, 28040, Madrid, Spain}
\email{tomcooney1@gmail.com}

\author{M. Junge}
\address{Department of Mathematics, University of Illinois, Urbana, IL 61801, USA}
\email{junge@math.uiuc.edu}

\author{C. Palazuelos}
\address{Instituto de Ciencias Matem\'aticas, CSIC-UAM-UC3M-UCM, CSIC, 28049, Madrid. Spain}
\email{carlospalazuelos@icmat.es}

\author{D. P\'erez-Garc\'ia}
\address{Departamento de An\'alisis Matem\'atico and IMI, Universidad Complutense de
Madrid, 28040, Madrid, Spain}
\email{dperez@mat.ucm.es}

\thanks{The first, third and fourth authors are partially supported by the Spanish grants QUITEMAD, I-MATH, MTM2011-26912, S2009/ESP-1594 and the European project QUEVADIS. The second author is partially supported by NSF DMS-0901457. The third author is partially supported by ``Juan de la Cierva'' program (Spain)}

\maketitle

\begin{abstract}
In this work we study rank-one quantum games. In particular, we focus on the study of the computability of the entangled value $\omega^*$. We  show that the value $\omega^*$ can be efficiently approximated up to a multiplicative factor of $4$. We also study the behavior of $\omega^*$ under the parallel repetition of rank-one quantum games,  showing that it   does not verify a perfect parallel repetition theorem. To obtain these results, we  first connect rank-one games with the mathematical theory of operator spaces. We  also reprove with these new tools essentially known results about the entangled value of rank-one games with one-way communication $\omega_{qow}$. In particular, we  show that $\omega_{qow}$ can be computed efficiently and it satisfies a perfect parallel repetition theorem.
\end{abstract}

\section{Introduction}\label{sec introduction}

The study of two-player one-round games is a central topic in both theoretical computer science and quantum information theory (QIT). In theoretical computer science, they play a key role in analyzing  the complexity of approximating some combinatorial optimization problems. As for quantum information, two-player one-round games are a natural setting in which to understand Bell inequalities. Bell inequalities have always played a fundamental role in QIT and their applications cover a huge variety of topics, from cryptography to foundational issues.

A two-player one-round game $G$ is specified by a referee, who chooses a pair of questions according to a probability distribution and who sends one question to each of the players. These players respond with answers taken from a certain finite set. The referee decides whether the players win according to a predicate which depends on the questions and answers. The players can agree in advance on a strategy for their answers but they are not allowed to communicate with each other once the game has started. Computer scientists are mainly interested in the \emph{classical value} of the game, $\omega(G)$, which is defined as the maximum attainable winning probability of the players when they are allowed to use classical strategies. However, having in mind that quantum mechanics provides us with, in principle, new  possibilities, one can consider the maximum attainable winning probability of the players when they are allowed to share (unlimited) entanglement to define their strategies. We
  then talk of the \emph{entangled value} of the game $G$ and we denote it by $\omega^*(G)$. One of Bell's fundamental observations can be reformulated as saying that $w^*(G)\gl \sqrt{2}w(G)>w(G)$ for certain games $G$ (so the players can indeed define strictly better strategies if they are allowed to use quantum resources instead of just classical strategies). Moreover, the value $\omega^*(G)$ and the quotient $\frac{\omega^*(G)}{\omega(G)}$ have been shown to be very important parameters in QIT. Thus, the fundamental questions about $\omega(G)$  have also been studied for the entangled value of the game $\omega^*(G)$. In particular, the study of \emph{how hard it is to compute or approximate the value $\omega^*(G)$} and the behavior of the parameter $\omega^*(G)$ with respect to \emph{the parallel repetition of $G$} has captured the attention of many authors in the few last years (\cite{KKMTV}, \cite{CSUU}, \cite{KeRe}, \cite{KRT}, \cite{KeVi}).

It turns out that another class of games naturally arises in the context of quantum information. \emph{Quantum games} are those in which the communication between the referee and the players (the questions and answers) is transmitted using  quantum states. Specifically, in a quantum game the referee prepares an initial tripartite state $ABC$ and sends registers $A$ and $B$ to Alice and Bob (the players), respectively. After following their previously agreed strategy, Alice and Bob send back their new registers to the referee, who  tests the answers via a two outputs (win/lose) projective measurement. We will give a more precise explanation of these games in Section \ref{sec connections}. In the same way as above, one can define different values for quantum games according to the strategies that Alice and Bob are allowed to use. Since everything considered in these games is quantum, the classical value of the game $\omega(G)$ does not seem so natural. However, we can define the \emph{
 entangled valued} of a quantum game, $\omega^*(G)$, in exactly the same way as before. That is, $\omega^*(G)$ is the maximum attainable winning probability of the players when they are allowed to share unlimited entanglement to define their strategies. Furthermore, we can also consider the value of the game $V(G)$ when Alice and Bob are allowed to share unlimited entanglement and to send unlimited amount of two way quantum information. $V(G)$ is called the {\it maximal value of the game} and trivially coincides with the maximum attainable winning probability for a unique player with access to both Alice and Bob's Hilbert spaces. It was recently shown by Buhrman et al. (\cite[Theorem 4.1]{BCFGGOS}),  that in order to obtain the value $V(G)$ it is enough to consider the \emph{entangled value of the game with simultaneous mutual communication}. That is, when Alice and Bob share unlimited entanglement and both can send an unlimited amount of quantum information, with the restriction tha
 t  their messages cannot depend on the ones received. As an intermediate situation, we may also consider the case of \emph{one-way communication}. This means that we  allow one of the players to send information to the other one, but not the other way around. We  talk in this case about the \emph{entangled value of the game with one-way communication} and we  denote it by $\omega_{qow}(G)$. As we will show (see Section \ref{Section: three models}) the values $V(G)$, $\omega_{qow}(G)$ and $\omega^*(G)$ can be very different for certain rank-one quantum games $G$.

Two recent papers have studied quantum games from different perspectives (\cite{KKMTV}, \cite{LTW}). In the first work, the authors studied general quantum games. Following the approach mentioned before, they studied some important parameters, which arise in the context of computer science, for the entangled value of quantum games. In particular, in \cite{KKMTV} Kempe at al. studied the hardness of computing the value $\omega^*(G)$. One of the main results presented in that work states that \emph{it is NP-hard to approximate the entangled value of a general quantum game with inverse polynomial precision}. On the other hand, the approach followed in \cite{LTW} by Leung et al. was via studying some particular quantum games. Indeed, motivated by the study of how much entanglement is needed to optimally play a quantum game, in \cite{LTW} the authors considered a particular case of quantum games, the so called \emph{coherent state exchange games}. Then, the authors showed that some of the
 se quantum games can be played optimally with an infinite amount of entanglement (that is, $\omega^*(G)=1$), though no finite dimensional entangled state can define a strategy that wins with probability one. It is very interesting to mention that there is not any known analogous result for the entangled value of classical games. The reader can find some other references on quantum games considering some other problems: a single player (\cite{Wat2}, \cite{KW}), limited prior entanglement (\cite{KoMa}) or players using  classical communication but with no prior shared entanglement (\cite{BAP}).

In this paper we deal with those quantum games in which the projective measurement of the referee is defined by a rank-one projection. We call these games \emph{rank-one quantum games}. One example of rank-one quantum games are the coherent state exchange games studied by Leung et al. in \cite{LTW}. We will also introduce some other examples of rank-one quantum games that have some interesting properties. Our approach to the study of these games is via \emph{operator spaces}. Operator space theory can be understood as a non-commutative version of  Banach space theory and have been shown to be a natural mathematical tool in quantum information theory. In recent years, they have been applied in several contexts like Bell inequalities (\cite{PWJPV}, \cite{JPPVW2}, \cite{JNPPSW}, \cite{JP}), quantum channels (\cite{DJKR}, \cite{JKP}), and entanglement theory  (\cite{JKPP}). We also refer to \cite{PisierSurvey} for a very nice survey on the topic. The main connection established in this work says that \emph{given a rank-one quantum game $G$, the entangled value $\omega^*(G)$ and the entangled value with one-way communication $\omega_{qow}(G)$ can be expressed by certain operator space norms on the tensor product of two matrix spaces, one corresponding to Alice and the other to Bob}\footnote{The formal result is Theorem \ref{main- connection} in Section \ref{sec connections}.}. With this connection at hand, we are able to study both problems: the hardness of computing or approximating the value $\omega^*(G)$ as well as its behavior with respect to the parallel repetition of the game. Along  the way we also recover the corresponding results for $\omega_{qow}(G)$, which were essentially already known before  albeit using completely different techniques. We note that Rapaport and Ta-Shma describe a protocol in Section 3.3 of \cite{RA} which is basically the same as that defining a rank-one quantum game. Their Theorem 3.2 gives a formula for $\omega^*(G)$. In order to avoid confusion, we point out that ``rank-one'' has a different meaning in this work than in their paper. On the other hand, during the process of preparing this work we learnt that Regev and Vidick came up with similar connections while considering different kinds of games \cite{RegevVidick}. Although there is a connection between both works (see \cite[Section 5.1]{RegevVidick}) most of the questions, so the results, considered there are different form those treated in our work.

We strive to make this paper as accessible as possible; our presentation is aimed at readers who are not already familiar with operator spaces. There is a section devoted to the basic definitions and results from this theory; we do not attempt to give an overview of the field and only include those results that are needed in this work. In order to make this paper more self-contained, we provide simple derivations of some estimates that are well-known to specialists in operator spaces. If we do not reprove a result (like the operator space Grothendieck inequality, Theorems \ref{Grothendieck} and \ref{Grothendieck II}), we discuss how the version of this theorem applied to quantum games in this paper can be derived from the more general theorems appearing in the operator space literature.

\subsection{Summary of results}

\subsubsection{Computing and approximating the entangled value}

Since the most important parameters in the study of games are the corresponding values (classical, entangled, one-way, \dots), it is natural to ask how hard these values are to compute or approximate. In the setting of classical games the problem is quite well understood if we focus on its classical value. Indeed, as a consequence of the PCP theorem (\cite{ALMSS}, \cite{ArSa}) and the parallel repetition theorem of Raz \cite{Raz}, one can deduce that, unless P=NP, for any fixed $\epsilon>0$ there is no algorithm working in polynomial time in the number of questions and answers which can decide whether the value of a two-player one-round game is $1$ or $<\epsilon$, given the promise that one of the two options happens.

Surprisingly, up to some results on particular kinds of games, much less is known about the computability or approximability of the entangled value of a general two-player one-round classical game. In \cite{IKM} the authors proved that it is NP-hard to approximate the entangled value of a two-player one-round classical game, $\omega^*(G)$, with inverse polynomial precision. Regarding positive results, the only known cases are XOR games, whose entangled value can be efficiently computed (\cite{CSUU})\footnote{Interestingly, Vidick has proved very recently that for any $\epsilon > 0$ the problem of finding a factor $(2 - \epsilon)$ approximation to the entangled
value of a three-player XOR game is NP-hard (\cite{Vidick}).} and with unique games (a more general class than XOR games), whose entangled value can be efficiently approximated, at least when this value is very close to one (\cite{KRT}). Regarding quantum games, however, it was proved in \cite{KKMTV} that it is NP-hard to approximate the entangled value $\omega^*(G)$ with inverse polynomial precision. In this paper we study the computability of $\omega^*(G)$  for rank-one quantum games $G$. Our main theorem states as follows.
\begin{theorem}\label{Main I}
The entangled value $\omega^*$ can be efficiently approximated up to a multiplicative constant relative error of $4$ on rank-one quantum games.
\end{theorem}
Notice that, in contrast to the main result in \cite{KKMTV} (which proves the NP-hardness result of approximating $\omega^*(G)$ up to an inverse polynomial for general quantum games), Theorem \ref{Main I} shows how to approximate the value $\omega^*(G)$ for rank-one quantum games up to a multiplicative constant. As we will show, the approximability result on $\omega^*(G)$ is based on a deep theorem in operator space theory which is a \emph{non-commutative} version of Grothendieck's theorem. As far as we know, Theorem \ref{Main I} is the first result giving a positive result in this direction. Since the proof in \cite{KKMTV} does not apply to rank-one quantum games, we do not know if $\omega^*(G)$ can be efficiently computed or even approximated to polynomial precision. With not much extra effort, we will reprove with operator space ideas the {\it known} result that $\omega_{qow}$ can be efficiently computed.
\subsubsection{Parallel repetition of the game}

One of the most important problems in the study of classical games is how to amplify the gap: what procedure can decrease the value of games with value less than $1$ but without altering the value of those games that initially had value 1?
 This can be easily done if one allows to increase the number of rounds and/or the number of players (just repeating the game sequentially and/or repeating the game in parallel with independent pairs of players). However, the problem becomes much more difficult if we want to decrease the value of the game while using the same number of rounds and players. The natural way to do this is to repeat the game many times in parallel. That is, in the setting of classical games the referee chooses $n$ pairs of questions independently and sends to each player the corresponding $n$-tuple of questions. Then, each player responds with a $n$-tuple of answers, which are accepted if each of the $n$ answer pairs would have been accepted in the original game. If we denote by $G^n$ the game played $n$ times in parallel, it is trivial to see that
\begin{align*}
\omega(G^n)\geq \omega(G)^n.
\end{align*}Somewhat surprisingly, the previous inequality is in general a strict inequality (see \cite{Feige}). The problem of parallel repetition is then to find good upper bounds for the value $\omega(G^n)$. A long series of works on this problem culminated with the work of Raz (\cite{Raz}), where he proved the \emph{parallel repetition theorem}. That is, the value of a game repeated in parallel decreases exponentially with the number of repetitions $n$ (although not exactly at rate $\omega(G)^n$). We say that certain games verify a \emph{perfect parallel repetition theorem} if $\omega(G^n)= \omega(G)^n$ for every $n$. Note that this problem can also be stated exactly in the same way for the entangled value of a classical game $\omega^*(G)$. However, in this context the situation is not so well understood. In \cite{CSUU} the authors showed a perfect parallel repetition for the entangled value of XOR games. After that, in \cite{KRT} the authors proved a parallel repetition theorem
 for unique games. Regarding the general situation, the best known result was given in the very recent work \cite{KeVi}, where the authors showed that the value $\omega^*(G)$ can be indeed reduced through parallel repetition, provided it was not initially $1$. The best rate of decrease for the value $\omega^*(G)$ obtained by repeating the game is still an open problem.

In this work we will study the parallel repetition of a rank-one quantum game. Given a rank-one quantum game $G$, one can analogously define a parallel repetition of this game $G^n$ just by considering the tensor product of both the preparation state and the rank-one projection which defines the referee's test (in particular, a parallel repetition of a rank-one quantum game is again a rank-one quantum game). We will study here whether there exists a perfect parallel repetition for the value $\omega^*(G)$ on rank-one quantum games. We will show
\begin{theorem}\label{Main II}
The entangled value does not verify a perfect parallel repetition theorem on rank-one quantum games. Specifically, for every natural number $n$ there exists a rank-one quantum game $G$ of local dimension $n$ for which $$\frac{\omega^*(G^2)}{\omega^*(G)^2} \succeq n^2,$$where $\succeq$ denotes inequality up to a universal constant independent of $n$.
\end{theorem}
Having in mind that we have a perfect parallel repetition theorem for the entangled value of XOR classical games, the second result is somehow surprising. It says that even in the most basic scenario of quantum games, the rank-one quantum games, we do not have a perfect parallel repetition theorem. Motivated by this fact, we will present a quite large family of rank-one quantum games for which perfect parallel repetition theorem is not far from being true. However, we will show that even for those games perfect parallel repetition fails. Again, as before, we also use our techniques to recover the known statement that $\omega_{qow}$ satisfies a perfect parallel repetition theorem.

To conclude, we mention that the techniques used in this paper also apply to the general case of quantum games. In particular, one can describe general quantum games via certain tensor norms in the framework of operator spaces. However, since the results for the general case require even more of the technology of tensor products, we defer the treatment of general quantum games to a forthcoming paper.

\subsubsection{New examples of quantum games}

We introduce families of games that clearly demonstrate the relative power of various quantum resources and which prove that the three values $\omega^*$, $\omega_{qow}$, and $V$ describe genuinely different scenarios. The restriction to rank-one projective measurements makes these games easier to study but not at a too high a price;  rank-one quantum games remain a wide enough class of games to be useful in studying the power of these resources. The games $G_C$ show that quantum one-way communication provide a greater advantage than shared entanglement; we construct games $G_C$ with an arbitrarily large separation between the values $\omega^*(G_C)$ and $\omega_{qow}(G_C)$. Similarly, the games $G_R$ show that quantum one-way communication is less powerful than the two-way quantum communication that allows one to attain the maximal value $V(G_R)$. We prove that the games $G_R$ and $G_C$ yield a separation between these values that is optimal in the dimensions of the player's Hilbert s
 paces.

We introduce two more families of games which we call Schur games and $OH_n$-games. Schur games satisfy a seemingly very restrictive condition and yet remain rich enough to demonstrate interesting features of entanglement. The coherent state exchange games of \cite{LTW} discussed above are Schur games. Despite their apparent simplicity, Schur games indeed form a nontrivial class of games; we provide examples of Schur games whose maximal and entangled values differ by an arbitrarily large multiplicative constant. We also show that, in contrast to the case of general rank-one quantum games, quantum one-way communication provides only a bounded advantage over the use of entanglement  for Schur games. We then return to the study of the parallel repetition of rank-one quantum games; the $OH_n$-games provide examples of games whose entangled values do not satisfy perfect parallel repetition theorems but which still satisfy strict bounds on the growth of $\omega^*(G^k)$.

\subsection{Structure of the paper}

The paper is organized as follows. Section \ref{sec operator spaces} is devoted to introducing the basic tools that we will need about operator spaces. In order to make this work more accessible to its intended audience, we reprove many previously known estimates and explain how the versions of results that we will apply to quantum games can be obtained from those already appearing in the operator space literature. In the first part of Section \ref{sec connections} we will explain in detail quantum games and the values $\omega^*$, $\omega_{qow}$; we will show the connections between these values and certain tensor norms in the category of operator spaces. In Section \ref{Section: three models}, we discuss some particularly interesting examples of rank-one quantum games. We use these games to show the existence of arbitrarily large gaps between $\omega^*$, $\omega_{qow}$, and the maximal value $V$, demonstrating that these three values do indeed describe different scenarios. Sections
 \ref{Sec approx} and \ref{sec parallel repetition} will be devoted, respectively, to proving Theorem \ref{Main I} and Theorem \ref{Main II}.
In Section \ref{Section:SchurOHn} we introduce two new families of quantum games that can be effectively studied using operator space techniques.


\section{Operator Spaces}\label{sec operator spaces}

In this section we introduce some basic concepts from operator space theory. We focus only on those aspects which are useful for this work and we direct the interested reader to the standard references \cite{EffrosRuan}, \cite{Pisier}. With this in mind, we include a few simple proofs of certain estimates, well-known to specialists in operator spaces, which are at the heart of many results in this work. In the following, given Hilbert spaces $\mathcal H$ and $\mathcal K$, we will denote by $B(\mathcal H, \mathcal K)$ the space of bounded operators from  $\mathcal H$ to $\mathcal K$ endowed with the standard operator norm. In this work, $\C^n$ is always endowed with its Hilbertian norm $\|\sum_{i=1}^n\alpha_i|i\rangle\|=\big(\sum_{i=1}^n|\alpha_i|^2\big)^\frac{1}{2}$. When $\mathcal H=\C^n$ and $\mathcal K=\C^m$ we will denote $M_{n,m}=B(\C^n, \C^m)$ and in the case where $n=m$ we will just write $M_n$.
\subsection{Operator spaces and completely bounded maps}
An \emph{operator space} $V$ is a complex Banach space together with a sequence of \emph{matrix norms} $\N{\cdot}_n$ on $M_n(V)$ satisfying the following conditions:
\begin{itemize}\label{OS properties}
\item $\N{v \oplus w}_{m+n} = \max\{\N{v}_m,\N{w}_n \}$ and
\item $\N{\alpha v \beta}_n \leq \N{\alpha} \, \N{v}_m \, \N{\beta}$
\end{itemize}
for all $v \in M_m(V)$, $w \in M_n(V)$, $\alpha \in M_{n,m}$, and $\beta \in M_{m,n}$.

In order to understand this theory, one also needs to study the morphisms that preserve the operator space structure. In contrast to Banach space theory, where one needs to study the bounded maps between Banach spaces, in the theory of operator spaces we need to study the \emph{completely bounded maps}. Given operator spaces $V$ and $W$ and a linear map $T:V \to W$, let $T_n: M_n(V) \to M_n(W)$ denote the linear map defined by
\[
T_n(v)=(id_n \otimes T)(v)=(T(v_{ij}))_{i,j}.
\]
A map is said to be \emph{completely bounded} if
\[
\N{T}_{cb} = \sup_n \N{T_n}< \infty,
\]
and this quantity is then called the \emph{completely bounded} norm of $T$. It is not difficult to see that $\|T^*\|_{cb}=\|T\|_{cb}$ for every $T:V \to W$, where $T^*$ denotes the adjoint map of $T$. We will say that $T$ is \emph{completely contractive} if $\|T\|_{cb}\leq 1$. Moreover, $T$ is said to be a \emph{complete isomorphism} (resp. \emph{complete isometry}) if each map $T_n$ is an isomorphism (resp. isometry). In particular, given two operator spaces $V$ and $W$, we define the \emph{completely bounded Banach-Mazur distance} between them as
\begin{align}\label{B-M distance}
d_{cb}(V,W)=\inf\Big\{\|T\|_{cb}\|T^{-1}\|_{cb}: T:V\rightarrow W \text{   }\text{   is an isomorphism}\Big\}.
\end{align}
A simple, but important, example of an operator space is $M_N$ with its operator space structure given by the usual sequence of matrix norms $\N{\cdot}_n$ defined by the identification $M_n(M_N)=M_{nN}$. Given a linear map $T:M_N\rightarrow M_N$, we can then compute its completely bounded norm
$$\|T\|_{cb}=\sup_n \|(id_n \otimes T):M_{nN}\rightarrow M_{nN}\|=\|(id_N \otimes T):M_{N^2}\rightarrow M_{N^2}\|,$$
where the last equality is a well known result proved by Smith (\cite{Smith}). In fact, it can be seen that given a linear map $T:M_N\rightarrow M_N$, $$\|T\|_{cb}=\|T^*\|_\diamondsuit,$$where $\|\cdot\|_\diamondsuit$ denotes the \emph{diamond norm} introduced in \cite{Kitev} and already used in many different contexts in complexity theory and information theory.
\subsection{The Column and Row structures}

In general, a Banach space can be endowed with different operator space structures that are not completely isometrically isomorphic to each other. We now describe two operator space structures on $\C^N$ which will play a central role in this paper. The \emph{Column} and \emph{Row} Hilbertian operator space structures are defined, respectively, by the sequences of matrix norms
\begin{align}\label{CR structures}
\big\|\sum_{i=1}^NA_i\otimes |i\rangle\big\|_{M_n(\C^N)}=\big\|\sum_{i=1}^NA_i^*A_i\big\|^\frac{1}{2}_{M_n}, \text{    } \text{    }
\big\|\sum_{i=1}^NA_i\otimes |i\rangle\big\|_{M_n(\C^N)}=\big\|\sum_{i=1}^NA_iA_i^*\big\|^\frac{1}{2}_{M_n}.
\end{align}
We will denote by $C_N$ (resp. $R_N$) the space $\C^N$ endowed with its Column (resp. Row) operator space structure. It is easy to deduce from the previous definition that these structures do not depend on the basis $(|i\rangle)_{i=1}^N$ chosen. The computation of the completely bounded norm of a linear map $T:C_N\rightarrow R_N$ is particularly easy, in contrast to the general case. Indeed, the following well-known lemma tells us how to compute such a norm.
\begin{lemma}
Let $T:\C^N\rightarrow \C^N$ be a linear map. Then,
\begin{align}
\|T:C_N\rightarrow R_N\|_{cb}=\big(\sum_{i=1}^N|\lambda_i|^2\big)^\frac{1}{2}=\|T:R_N\rightarrow C_N\|_{cb},
\end{align}
where $T=\sum_{i=1}^N\lambda_i|e_i\rangle\langle f_i|$ is the singular value decomposition of $T$.
\end{lemma}
As the above lemma is key to many results in this paper, we include the proof.
\begin{proof}
We will only prove the first equality since the second one can be proved in the same way. By the comments above, we can assume that $T=\sum_{i=1}^N\lambda_i|i\rangle\langle i|$.

Let us denote $A_i=|1\rangle\langle i|\in M_N$ for every $i=1,\cdots, N$. Then, according to Equation (\ref{CR structures}) we trivially have $\big\|\sum_{i=1}^NA_i\otimes |i\rangle\big\|_{M_n(C_N)}\leq 1$. Therefore,
\begin{align*}
\Big\|(id_N \otimes T)\Big(\sum_{i=1}^NA_i\otimes |i\rangle\Big)\Big\|_{M_N(R_N)}=\big\|\sum_{i=1}^N\lambda_iA_i\otimes |i\rangle\big\|_{M_N(R_N)}\leq \|T\|_{cb}.
\end{align*}However, note that $\big\|\sum_{i=1}^N\lambda_iA_i\otimes |i\rangle\big\|_{M_N(R_N)}=(\sum_{i=1}^N|\lambda_i|^2\big)^\frac{1}{2}$. Therefore, $\|T\|_{cb}\geq \big(\sum_{i=1}^N|\lambda_i|^2\big)^\frac{1}{2}$.

To see the converse inequality, let us consider a sequence of matrices $(A_i)_{i=1}^N\subset M_n$ such that $\big\|\sum_{i=1}^NA_i\otimes |i\rangle\big\|_{M_n(C_N)}=\big\|\sum_{i=1}^NA_i^*A_i\big\|^\frac{1}{2}_{M_n}\leq 1$. Then,
\begin{align*}
\Big\|(id_N \otimes T)\Big(\sum_{i=1}^NA_i\otimes |i\rangle\Big)\Big\|_{M_n(R_N)}=\Big\|\sum_{i=1}^N\lambda_iA_i\otimes |i\rangle\Big\|_{M_n(R_N)}\\=\Big\|\sum_{i=1}^N|\lambda_i|^2A_iA_i^*\Big\|^\frac{1}{2}_{M_n}\leq \big(\sum_{i=1}^N|\lambda_i|^2\big)^\frac{1}{2},
\end{align*}where we have used that $\|A_iA_i^*\|\leq 1$ for every $i$.
\end{proof}
A direct consequence of the previous lemma is that
\begin{align*}
\|id:R_N\rightarrow C_N\|_{cb}=\|id:C_N\rightarrow R_N\|_{cb}=\sqrt{N}.
\end{align*}
Moreover, one can even deduce that
\begin{align}\label{RC Banach-Mazur}
d_{cb}(C_N, R_N)=N
\end{align}Indeed, the upper bound follows trivially by considering $T=id$ in (\ref{B-M distance}). On the other hand, whenever we have an isomorphism $u:C_N\rightarrow R_N$ we will have
\begin{align*}
N=tr(id)=tr(u^{-1}\circ u)\leq \big(\sum_{i=1}^N|\beta_i|^2\big)^\frac{1}{2}\big(\sum_{i=1}^N|\lambda_i|^2\big)^\frac{1}{2}\\=\|u^{-1}:R_N\rightarrow C_N\|_{cb}\|u:C_N\rightarrow R_N\|_{cb},
\end{align*}where we have used Cauchy-Schwartz inequality and the singular value decompositions $u=\sum_{i=1}^N\lambda_i|e_i\rangle\langle f_i|$,  $u^{-1}=\sum_{i=1}^N\beta_i|e_i\rangle\langle f_i|$.

We encourage the reader to prove the following well known estimates (which can be found in \cite{Pisier})
\begin{align}\label{distance R-C-M_N}
\sqrt{N}=&\|id:M_N\rightarrow R_{N^2}\|_{cb}=\|id:R_{N^2}\rightarrow M_N\|_{cb}\\\nonumber=&
\|id:M_N\rightarrow C_{N^2}\|_{cb}=\|id:C_{N^2}\rightarrow M_N\|_{cb}.
\end{align}
Given a linear map between two operator spaces $T:V\rightarrow W$ we will define
\begin{align*}
\Gamma_R(T)=\inf\big\{\|a:V\rightarrow R_n\|_{cb}\|b:R_n\rightarrow W\|_{cb}\big\},
\end{align*}where the infimum runs over all $n$ and all possible factorizations $T=b\circ a$, where $a$ and $b$ are linear maps. It is trivial to check that $\|T\|_{cb}\leq \Gamma_R(T)$ for every linear map $T$ and that
\begin{align}\label{R fact R}
\Gamma_R(T:V\rightarrow W)=\|T\|_{cb} \text{     } \text{     }\text{     if $V$ or $W$ is equal to $R_N$ for some $N$.}
\end{align}On the other hand,
\begin{align}\label{R fact CC}
\Gamma_R(id:C_N\rightarrow C_N)=N \text{     } \text{    and     }\text{    } \Gamma_C(id:R_N\rightarrow R_N)=N
\end{align}
is a trivial consequence of (\ref{RC Banach-Mazur}).

Note that the norms defined in (\ref{CR structures}) can be understood in the following way. Let us identify the space $\C^N$ with the linear space described by the first column (resp. row) of the matrices in $M_N$. That is,
\begin{align*}
C_N = \textrm{span}\ \{|i\rangle\langle 1| \mid 1 \leq i \leq N\} \subseteq M_N, \text{    } \text{    }R_N = \textrm{span}\ \{|1\rangle\langle i| \mid 1 \leq i \leq N\} \subseteq M_N.
\end{align*}Then, one can easily check that the norms defined in (\ref{CR structures}) are nothing else than the corresponding inherit norms on $M_n(M_N)$. This leads us to an equivalent definition of operator spaces, as those subspaces of $B(\mathcal H)$. Note that, given a subspace $V\subset B(\mathcal H)$ we have a family of matrix norms, by identifying $M_n(V)\subset M_n(B(\mathcal H))=B(\C^n\otimes \mathcal H)$, which can be shown to be an \emph{acceptable} sequence of matrix norms. The converse statement is known as Ruan's Theorem and can be found in \cite[Theorem 2.3.5]{EffrosRuan}.
\begin{remark}\label{tensor product structures I}
It is very easy to see that the operator space structure on $R_N\otimes R_N$ (resp. $C_N\otimes C_N$) as a subspaces of $M_{N^2}=M_N\otimes M_N$ is $R_{N^2}$ (resp. $C_{N^2}$). On the other hand, one can also see that the operator space structure on $R_N\otimes C_N$ and $C_N\otimes R_N$ as a subspaces of $M_{N^2}$ is $M_N$ in both cases. Moreover, in all these cases there exist completely bounded norm one projections onto the included subspaces.
\end{remark}
As in Banach space theory, we can also consider the notion of duality. Given an operator space $V$, we define the \emph{dual operator space} $V^*$ by means of the matrix norms
\begin{align*}
M_n(V^*)=CB(V,M_n), \text{    } n\geq 1,
\end{align*}where the identification associates to any element $x=\sum_ia_i\otimes v_i^*\in M_n(V^*)=M_n\otimes V$ the linear map defined by $T_x(v)=\sum_i v_i^*(v)a_i$ for every $v\in V$.
With this definition it is a simple exercise to show that the following identifications are complete isometries:
\begin{align*}
C_N^* =R_N, \qquad R_N^*=C_N.
\end{align*}
If we denote by $S_1^N$ the space $M_N$ with the trace norm, we can then define a natural operator space structure on it via the duality relation $S_1^N=M_N^*$. Note that this operator space structure is \emph{not} given by the linear map identifying matrices in $S_1^N$ with matrices in $M_N$ as this map does not induce the correct norm on $S_1^N$. By \cite[Theorem 3.2.3]{EffrosRuan}, the \emph{scalar pairing}
\begin{align}\label{Scalar Pairing}
\langle b, c\rangle =\sum_{s,t} b_{s,t}c_{s,t}=\mathrm{tr}\, (bc^{tr})
\end{align}
yields completely isometric isomorphisms $M_N^* = S_1^N$ and $(S_1^N)^* = M_N$.

Although we will not use it explicitly in this work, one can prove (\cite[Theorem 5.3]{Pisier2}) that the sequence of matrix norms defining the previous operator space structure on $S_1^N$ is given by
\begin{align*}
\Big\|\sum_{i,j=1}^Nx_{i,j}\otimes |i\rangle\langle j|\Big\|_{M_n(S_1^N)}=\sup\Big\{\Big\|\sum_{i,j=1}^Nax_{i,j}b\otimes |i\rangle\langle j|\Big\|_{S_1^{nN}}\Big\},
\end{align*}where the supremum runs over all matrices $a,b\in M_n$ with Hilbert Schmidt norm lower than or equal to one. However, it can be seen that this norm coincides with the completely bounded norm of the map $T:M_N\rightarrow M_n$ defined by
\begin{align*}
T(A)=\sum_{i,j=1}^N\langle i|A|j\rangle x_{i,j}.
\end{align*}Therefore, in this case one can ``easily'' define the family of matrix norms on $M_n(S_1^N)$. Note, however, that we are not giving the explicit form of the embedding of $S_1^N$ in $B(\mathcal H)$ as we did for the Column and Row structures on $\C^n$. In fact, finding ``good embedding'' is in general a very tough problem.

The previous duality relation together with (\ref{distance R-C-M_N}) guarantees that
\begin{align}\label{distance R-C-S_1^N}
\sqrt{N}=&\|id:S_1^N\rightarrow R_{N^2}\|_{cb}=\|id:R_{N^2}\rightarrow S_1^N\|_{cb}\\\nonumber=&
\|id:S_1^N\rightarrow C_{N^2}\|_{cb}=\|id:C_{N^2}\rightarrow S_1^N\|_{cb}.
\end{align}
If $P : M_N \to R_N$ is the projection onto the first row, then $P^*:C_N \to S_1^N$ provides a completely isometrically embedding of $C_N$ into $S_1^N$. Similarly, $R_N$ can be identified with the first column inside $S_1^N$ in a completely isometrically way. In fact, the corresponding projections are complete contractions. We can collect this information by saying that the following inclusions are completely complemented and completely isometric isomorphisms.
\begin{align}\label{CR in S_1}
C_N = \textrm{span}\ \{|1\rangle\langle i| \mid 1 \leq i \leq N\} \subseteq S_1^N, \text{     }\text{     }R_N = \textrm{span}\ \{|i\rangle\langle 1| \mid 1 \leq i \leq N\} \subseteq S_1^N.
\end{align}
\begin{remark}\label{tensor product structures II}
Remark \ref{tensor product structures I} and the above duality relations allow us to state that the operator space structure on $R_N\otimes R_N$ (resp. $C_N\otimes C_N$) as a subspaces of $S_1^{N^2}$ is $R_{N^2}$ (resp. $C_{N^2}$). On the other hand, one can also see that the operator space structure on $R_N\otimes C_N$ and $C_N\otimes R_N$ as a subspaces of $S_1^{N^2}$ is $S_1^N$ in both cases. Moreover, in all these cases there exist completely bounded norm one projections onto the specified subspaces.
\end{remark}
\subsection{The minimal and the Haagerup tensor norms}
Tensor norms will be very important in our work. In particular, we will need to introduce the minimal and the Haagerup tensor norms. Given two operator spaces $V \hookrightarrow B(\mathcal H_V)$ and $W \hookrightarrow B(\mathcal H_W)$, their algebraic tensor product $V \otimes W$ can be seen as a subspace of $B(\mathcal H_V \otimes \mathcal H_W)$ and their \emph{minimal operator space tensor product} $V \otimes_{min} W$ is the closure of $V \otimes W$ in $B(\mathcal H_V \otimes \mathcal H_W)$. An equivalent formulation, more useful for us, can be stated by saying that if $u =\sum_{i=1}^l v_i \otimes w_i\in V\otimes W$ we have
\begin{align*}
\|u\|_{V\otimes_{min}W}=\sup \Big\{\Big\|
\sum_{i=1}^l T(v_i) \otimes S(w_i)\Big\|\Big\}_{B(\mathcal H_V \otimes \mathcal H_W)},
\end{align*}where the supremum runs over all completely contractions $T:V\rightarrow B(\mathcal H_V)$, $S:V\rightarrow B(\mathcal H_W)$.
Moreover, it is easy to see that we can restrict to finite dimensional Hilbert spaces $\mathcal H_V$ and $\mathcal H_W$. It is also straightforward to see that this tensor norm is commutative and associative.

If $V$ and $W$ are finite dimensional, it is straightforward to check from the previous definition that we have the following isometric identification.
\begin{align}\label{min-operator norm}
V \otimes_{min} W = CB(V^*,W),
\end{align}where here the correspondence is defined by $\big(\sum_{i=1}^nv_i\otimes w_i\big)(v^*)=\sum_{i=1}^n\langle v_i,v^*\rangle w_i$.

A second fundamental tensor norm is the Haagerup tensor norm. Suppose that $V$ and $W$ are operator spaces. The \emph{Haagerup tensor product} norm of $u\in V\otimes W$ is defined as
\begin{align}\label{eq:Haagerup}
\|u\|_{V\otimes_h W}=\inf\Big\{\big\|\sum_{i=1}^k|i\rangle\otimes a_i\big\|_{R_k\otimes_{min} V}\big\|\sum_{i=1}^k|i\rangle\otimes b_i\big\|_{C_k\otimes_{min} W}\Big\},
\end{align}where the infimum runs over all possible ways to write $u$ as a finite sum $\sum_{i=1}^ka_i\otimes b_i\in V\otimes W$.
It is easy to see that this tensor norm is associative but not commutative. Indeed, one can define the transpose version of the Haagerup norm $h^t$ as
\begin{align*}
\|u\|_{V\otimes_{h^t} W}=\inf\Big\{\big\|\sum_{i=1}^k|i\rangle\otimes a_i\big\|_{C_k\otimes_{min} V}\big\|\sum_{i=1}^k|i\rangle\otimes b_i\big\|_{R_k\otimes_{min} W}\Big\},
\end{align*}
where the infimum runs again over all possible ways to write $u$ as a finite sum $\sum_{i=1}^ka_i\otimes b_i\in V\otimes W$.
According to (\ref{min-operator norm}), $\big\|\sum_{i=1}^k|i\rangle\otimes a_i\big\|_{R_k\otimes_{min} V}$ and $\|\sum_{i=1}^k|i\rangle\otimes b_i\big\|_{C_k\otimes_{min} W}$ can be understood as the completely bounded norm of the maps $T:V^*\rightarrow R_k$, defined by $T(x)=\sum_{i=1}^k\langle a_i,x\rangle|i\rangle$, and $S:R_k\rightarrow W$, defined by $S(y)=\sum_{i=1}^k\langle i|y\rangle b_i$, respectively. Thus, we see that the analogous isometric identification to (\ref{min-operator norm}) for the Haagerup tensor norm is given by
\begin{align}\label{Row-Haagerup norm}
V \otimes_{h} W &= \Gamma_R(V^*,W),\\
V \otimes_{h^t} W &= \Gamma_C(V^*,W). \nonumber
\end{align}
Remarkably, the Haagerup tensor norm can be stated in the following equivalent form: Given $u= \sum_{k=1}^l v_k \otimes w_k \in V \otimes W$ we have
\begin{align*}
\N{u}_h =\sup \Big\{\Big\|\sum_{k=1}^l T(v_k)S(w_k)\Big\|_{B(\mathcal H)} \Big\},
\end{align*}
where the supremum is taken over all Hilbert spaces $\mathcal H$ and completely contractions $T:V \to B(\mathcal H)$, $S:W \to B(\mathcal H)$.
We refer to \cite[Theorem 5.1]{Pisier} for the non trivial proof of the equivalence between both definitions. This norm has some ``magic'' properties with no analogy in the Banach space category. In particular, it is self-dual (\cite[Corollary 5.8]{Pisier}): For every finite dimensional operator spaces $V$ and $W$ we have that
\begin{align}\label{duality haagerup}
(V\otimes_h W)^*=V^*\otimes_h W^*\text{       } \text{isometrically}.
\end{align}
It is easy to check from their definitions (see \cite[Section 2.1]{Pisier} and \cite[Chapter 5]{Pisier} for details) that both norms verify the following property:
\begin{align}\label{metric mapping property}
\big\|T\otimes S:V\otimes_\alpha W\hookrightarrow V'\otimes_\alpha W'\big\|\leq \|T\|_{cb}\|S\|_{cb}
\end{align}for $\alpha=min$ or $\alpha=h$ and for all linear maps $T:V\rightarrow V'$ and $S:W\rightarrow W'$. Furthermore, both norms can be seen to be injective:
\begin{align}\label{injectivity}
T\otimes S:V\otimes_\alpha W\hookrightarrow V'\otimes_\alpha W \text{    }\text{    }\text{is an isometry}
\end{align}for $\alpha=min$ or $\alpha=h$, whenever $T$ and $S$ are completely isometries.

Equation (\ref{metric mapping property}) and the estimates given in (\ref{distance R-C-M_N}), (\ref{distance R-C-S_1^N}) give us the following estimates\footnote{In fact, the proof of some of these estimates is implicit in some of the proofs provided in Section \ref{Section: three models}.}:
\begin{align}\label{optimality gaps model}
\sqrt{N}&=\big\|id\otimes id:S_1^N\otimes_{min} S_1^N\rightarrow C_{N^2}\otimes_{min} S_1^N\big\|=\big\|id\otimes id:C_{N^2}\otimes_{h} S_1^N\rightarrow S_1^N\otimes_{h} S_1^N\big\|\\&
=\nonumber\big\|id\otimes id:S_1^N\otimes_{h} S_1^N\rightarrow R_{N^2}\otimes_{h} S_1^N\big\|=\big\|id\otimes id:R_{N^2}\otimes_{h} S_1^N\rightarrow S_1^{N^2}\big\|.
\end{align}
Finally, the reader should note that we have defined the previous tensor products $V\otimes_\alpha W$ as Banach spaces. However, the minimal and the Haagerup tensor products have a natural operator space structure; that is, one can define a natural sequence of matrix norms $M_n(V\otimes_\alpha W)$ in both cases. Since this structure will not play any role in most of the results given in this work we have preferred not to include it in this introduction. Some extra information will be provided in Section \ref{sec parallel repetition}, where the use of this structure will allow us to give a simple proof of a perfect parallel repetition theorem for the one-way quantum value of a rank-one quantum game. However, the interested reader can find more information on this in \cite{Pisier} and, in particular, check that the previous properties (\ref{duality haagerup}), (\ref{metric mapping property}), (\ref{injectivity}) are still true after being changed to read ``completely isometrically'', $
 \|T\otimes_\alpha S\|_{cb}\leq \|T\|_{cb}\|S\|_{cb}$, and ``complete isometry'' respectively.
\subsection{Grothendieck's inequality}
In the very particular case in which the operator space under consideration is $S_1^N$, there exists a close connection between the minimal and the Haagerup tensor norms. This result is known in the literature as the \emph{operator space Grothendieck inequality}\footnote{See the review \cite{PisierSurvey} for the history and applications of Grothendieck type inequalities.}. Let us first state the theorem in its standard form (albeit only for finite dimensional Hilbert spaces) and we will explain later how to obtain the different versions we are interested in. We refer the reader to \cite{PS}, \cite{HM}, \cite{PisierSurvey} for a more general statement of the theorem.
\begin{theorem}[\cite{PS}, \cite{HM}]\label{Grothendieck}
Given two finite dimensional Hilbert spaces $\He_A$, $\He_B$ and an element $M\in S_1(\He_A)\otimes S_1(\He_B)$, let us denote by $\widehat{M}:B(\He_A)\rightarrow S_1(\He_B)$ the corresponding linear map. Then, there exist states $f_1, f_2\in S_1(\He_A)$ and $g_1, g_2\in S_1(\He_B)$ such that for every $a\in B(\He_A)$ and $b\in B(\He_B)$ we have
\begin{align*}
|\langle \widehat{M}(a),b\rangle|\leq \|\widehat{M}\|_{cb}\big(f_1(aa^*)^\frac{1}{2}g_1(b^*b)^\frac{1}{2}+
f_2(a^*a)^\frac{1}{2}g_2(bb^*)^\frac{1}{2}\big).
\end{align*}
\end{theorem}
Theorem \ref{Grothendieck} was first proved by Pisier and Shlyakhtenko (\cite{PS}) with a constant factor $K$ in the right hand side of the inequality. Later, Haagerup and Musat gave a proof with $K=1$ (\cite{HM}). During the referee process of this work, a new proof of Theorem \ref{Grothendieck} was given by Regev and Vidick (\cite{RegevVidickII}). This proof, based on techniques from quantum information theory, improves certain quantitative estimates with respect to the previous proofs.

As it is explained in \cite[Lemma 3.1]{HM}, it follows from Theorem \ref{Grothendieck} that there is a decomposition of linear maps $\widehat{M}=\widehat{M}_1 + \widehat{M}_2$ and states $f_1,f_2 \in S_1(\h_A)$ and $g_1, g_2 \in S_1(\h_B)$ such that for all $a \in B(\h_A)$ and $b \in B(\h_B)$ we have
\begin{align}\label{u1Grothendieck}
|\langle\widehat{M}_1(a),b\rangle| \leq \|\widehat{M}\|_{cb} f_1(aa^*)^{1/2} g_1(b^*b)^{1/2}, \\
\nonumber |\langle\widehat{M}_2(a),b\rangle|\leq \|\widehat{M}\|_{cb} f_2(a^*a)^{1/2} g_2(bb^*)^{1/2}.
\end{align}Indeed, we are just identifying $\langle\widehat{M}(a),b\rangle=M(a,b)$, when $M$ is regarded as a bilinear form and using that the quantity $\N{u}_{ER}$ appearing in \cite{HM} satisfies $\N{u}_{ER} \leq \|\widehat{M}\|_{cb}$ by \cite[Proposition 3.3]{HM}. It is interesting to note that a standard Hahn-Banach argument (see \cite[Section 23]{PisierSurvey}) shows that condition (\ref{u1Grothendieck}) exactly means that
\begin{align*}
\|M_1\|_{S_1(\h_A)\otimes_h S_1(\h_B)}\leq \|\widehat{M}\|_{cb}=\|M\|_{S_1(\h_A)\otimes_{min} S_1(\h_B)},\\
\|M_2\|_{S_1(\h_A)\otimes_{h^t} S_1(\h_B)}\leq \|\widehat{M}\|_{cb}=\|M\|_{S_1(\h_A)\otimes_{min} S_1(\h_B)}.
\end{align*}
Now Theorem \ref{Grothendieck} can be stated in the following way:
\begin{theorem}\label{Grothendieck II}
Given two finite dimensional Hilbert spaces $\He_A$, $\He_B$ and an element $M\in S_1(\He_A)\otimes S_1(\He_B)$ we have
\begin{align}\label{constant}
\frac{1}{2}\|M\|_\mu\leq \|M\|_{S_1(\He_A)\otimes_{min} S_1(\He_B)} \leq\|M\|_\mu,
\end{align}where $\|\cdot\|_\mu$ is the so called {\em symmetrized Haagerup} tensor norm \cite{Pisier}, defined by
\begin{align*}
\|M\|_\mu=\sup \Big\{\langle M,A \rangle: \max\big\{\|A\|_{B(\He_A)\otimes_hB(\He_B)},\|A\|_{B(\He_A)\otimes_{h^t}B(\He_B)}\big\}\leq 1\Big\}.
\end{align*}
\end{theorem}
Let us explain how to obtain this result from the previous one. One first notices that $\|M\|_\mu$ can be rewritten as
\begin{equation*}
\|M\|_\mu=\inf \Big\{\|u\|_{S_1(\He_A)\otimes_{h} S_1(\He_B)}+\|v\|_{S_1(\He_A)\otimes_{h^t} S_1(\He_B)}: M=u+v\Big\}.\end{equation*}
To see why, let us call the second expression $\|M\|_*$ and assume it to be $< 1$. There must exist a decomposition $M=u+v$ with  $\|u\|_{S_1(\He_A)\otimes_{h} S_1(\He_B)}+\|v\|_{S_1(\He_A)\otimes_{h^t} S_1(\He_B)}< 1$. Let $A$ be any element in $B(\He_A)\otimes B(\He_B)$ verifying $$\max\big\{\|A\|_{B(\He_A)\otimes_hB(\He_B)},\|A\|_{B(\He_A)\otimes_{h^t}B(\He_B)}\big\}\leq 1.$$ One has
$$|\< M,A\>|\le |\<u,A\>|+|\<v,A\>|\le \|u\|_{S_1(\He_A)\otimes_{h} S_1(\He_B)}+\|v\|_{S_1(\He_A)\otimes_{h^t} S_1(\He_B)}< 1.$$Here, we have used the self duality of the Haagerup tensor norm and its transpose. This gives $$\|M\|_\mu\le \|M\|_*.$$
To see the other inequality we consider $M$ such that $\|M\|_*=1$. By the Hahn-Banach Theorem, there must exist $A$ with $\<M,A\>=1$ and $$|\<x,A\>|\le \|x\|_*\le \min\{\|x\|_{S_1(\He_A)\otimes_{h} S_1(\He_B)},\|x\|_{S_1(\He_A)\otimes_{h^t} S_1(\He_B)}\}$$ for all $x\in S_1(\He_A)\otimes S_1(\He_B)$. By the self duality of the Haagerup tensor norm and its transpose we get that $$\max\big\{\|A\|_{B(\He_A)\otimes_hB(\He_B)},\|A\|_{B(\He_A)\otimes_{h^t}B(\He_B)}\big\}\leq 1$$
and hence $\|M\|_\mu\ge 1$. This gives  $$\|M\|_\mu\ge \|M\|_*$$ finishing the argument.

\

Now, the second inequality (\ref{constant}) is trivial since, by (\ref{min-operator norm}) and (\ref{Row-Haagerup norm}), the minimal tensor norms is always smaller than the Haagerup one and its transpose. On the other hand, the first inequality in (\ref{constant}) follows trivially from (\ref{u1Grothendieck}) and the comments below it.
\section{Connections}\label{sec connections}
\subsection{Rank-one quantum games and connections to operator spaces}

We will start by explaining in detail rank-one quantum games. Actually, as we mentioned in Section \ref{sec introduction}, we will be interested in two different scenarios. In the first one, the two players, Alice and Bob, are allowed to use an entangled quantum state to define their strategy. In this case, the game works as follows:

\begin{enumerate}
\item The referee, Charlie, prepares an initial state $\ket{\psi} \in \mathcal H_A \otimes \mathcal H_B \otimes \mathcal H_C$ and sends the registers $A$ and $B$ to Alice and Bob, respectively.
\item Alice and Bob also share (an arbitrary amount of) entanglement in the form of a state $\ket{\phi} \in \mathcal H_{A'} \otimes \mathcal H_{B'}$. All Hilbert spaces are assumed to be finite-dimensional.
\item Alice and Bob apply quantum operations $T_{AA'}$ and $T_{BB'}$ to $\mathcal H_A \otimes \mathcal H_{A'} (\simeq \mathcal H_{AA'})$ and
$\mathcal H_B \otimes \mathcal H_{B'} (\simeq \mathcal H_{BB'})$ respectively.
\item Let $\ket{\gamma}$ be a state in $\mathcal H_A \otimes \mathcal H_B \otimes \mathcal H_C$. The triple $(A,B,C)$ is measured with respect to the projective measurement system with $P_1=\ket{\gamma}\bra{\gamma}$ and $P_0=1-P_1$. The outcome $1$ indicates that Alice and Bob win while $0$ means that they lose.
\item The game $G=G(\ket{\psi},\ket{\gamma})$ is completely determined by the initial state $\ket{\psi}$ and the measurement state $\ket{\gamma}$.
\item The value of the game is the supremum over all states $\ket{\phi}$ and quantum operations $T_{AA'}$ and $T_{BB'}$ of the probability that Alice and Bob win the game. This value will be denoted by $\omega^*(G)$ and it will be called \emph{entangled value} of $G$.
\end{enumerate}

Our second scenario is that in which Alice and Bob are allowed to transmit information in one direction. In principle they can share an entangled state too but this can be incorporated in the communication. In this case, as before, we assume that Alice, Bob, and Charlie share an initial state $\ket{\psi} \in \mathcal H_A\otimes \mathcal H_B \otimes \mathcal H_C$. After Alice and Bob perform their quantum operations, Charlie measures their responses against the measurement state $\ket{\gamma} \in \mathcal H_A \otimes \mathcal H_B \otimes \mathcal H_C$. However, in this situation, Alice is allowed to communicate with Bob; they communicate via an auxiliary system $A'$, initialized in the state $\ket{0} \in \mathcal H_{A'}$. Alice applies a quantum operation to $\mathcal H_A \otimes \mathcal H_{A'}$; Bob then applies a quantum operation to $\mathcal H_B \otimes \mathcal H_{A'}$. The value of the game is again the supremum over all quantum operations $T_{AA'}$ and $T_{BA'}$ of the probabi
 lity that Alice and Bob win the game. This value will be denoted by $\omega_{qow}(G)$ and it will be called \emph{entangled value of $G$ with one-way communication}.

The main result of this paper is that operator spaces are ideally suited to describing the value of such quantum games. By taking a partial trace over the referee's register, we obtain a matrix $M_{AB}=tr_C \ket{\psi}\bra{\gamma}$ in $S_1(\mathcal H_A) \otimes S_1(\mathcal H_B)$. Remarkably, by using different operator space tensor product norms, we can characterize those matrices corresponding to quantum games and describe the value of the quantum game corresponding to the different resources the players are allowed to use.

In order to pave the way for the main result, we will start by showing that $S_1(\mathcal H_A) \otimes S_1(\mathcal H_B)$ is indeed the natural space in which to realize the rank-one quantum games.
\begin{prop}\label{connection-projective}
Let $M \in S_1(\mathcal H_A \otimes \mathcal H_B)$. Then $\|M\|_{S_1(\mathcal H_A\otimes \mathcal H_B)} \leq 1$ if and only if there exist a finite dimensional Hilbert space $\mathcal H_C$ and $\ket{\psi}, \ket{\gamma}$ in the unit sphere of $\mathcal H_A \otimes \mathcal H_B \otimes \mathcal H_C$ such that $M=tr_C \ket{\psi}\bra{\gamma}$.
\end{prop}
\begin{proof}
It is very easy to see that $M=tr_C \ket{\psi}\bra{\gamma}$ is in the unit ball of $S_1(\mathcal H_A \otimes \mathcal H_B)$ for every states $\ket{\psi}, \ket{\gamma}\in\mathcal H_A \otimes \mathcal H_B \otimes \mathcal H_C$. To see the converse, let us consider the singular value decomposition of $M=\sum_{i=1}^N\alpha_i|f_i\rangle\langle g_i|$, where $(|f_i\rangle)_{i=1}^N$ and $(|g_i\rangle)_{i=1}^N$ are orthonormal bases of $\mathcal H_A \otimes \mathcal H_B$ and $(\alpha_i)_{i=1}^N$ is a sequence of nonnegative real numbers verifying $\sum_{i=1}^N\alpha_i\leq 1$. Then, we can consider $\mathcal H_C=\C^N$ and define $\ket{\psi}=\sum_{i=1}^N\sqrt{\alpha_i}|f_i\rangle|i\rangle$ and $\ket{\gamma}=\sum_{i=1}^N\sqrt{\alpha_i}|g_i\rangle|i\rangle$, where $(|i\rangle)_{i=1}^N$ is the canonical basis of $\C^N$. It is trivial that $\ket{\psi}$ and $\ket{\gamma}$ are in the unit ball of $\mathcal H_A \otimes \mathcal H_B \otimes \mathcal H_C$ and $M=tr_C \ket{\psi}\bra{\gamma}$. Finally, no
 te
 that we can assume that $\ket{\psi}$ and $\ket{\gamma}$ have both norm one just by considering  $\mathcal H_C=\C^{N+2}$ to complete norms.
\end{proof}
Proposition \ref{connection-projective} says that there is a one-to one correspondence between the unit ball of $S_1(\mathcal H_A\otimes \mathcal H_B)$ and the set of rank-one quantum games via the matrices $M_{AB}=tr_C \ket{\psi}\bra{\gamma}$. For any game $G=G(\ket{\psi},|\gamma\rangle)$, $V(G)=\|M_{AB}\|^2_{S_1(\mathcal H_A\otimes \mathcal H_B)}$ will be called \emph{the maximal value of the game}, since it represents the success probability of the game for one player that has access to both Hilbert spaces $\mathcal H_A$ and $\mathcal H_B$ (so the best possible situation). The main connection in this work states that the minimal tensor norm and the Haagerup tensor norm give us respectively the entangled value of the game and the entangled value of the game with one-way communication.
\begin{theorem}\label{main- connection}
Let $G$ be a quantum entangled game with initial state $\ket{\psi}$ and final measurement $\ket{\gamma}\bra{\gamma}$ for $\ket{\psi},\ket{\gamma} \in \mathcal H_A \otimes \mathcal H_B \otimes \mathcal H_C$. Let $M_{AB}=tr_C \ket{\psi}\bra{\gamma}$. Then,
\begin{enumerate}
\item[1.] The entangled value of the game is given by
\[
\omega^*(G)=\N{M_{AB}}_{S_1^A \otimes_{min} S_1^B}^2.
\]

\item[2.] The entangled value of the game with one-way communication from Alice to Bob is given by
\[
\omega_{qow}(G)=\N{M_{AB}}_{S_1^A \otimes_h S_1^B}^2.
\]
\item[3.] The maximal value of the game is given by
\[
V(G)=\N{M_{AB}}^2_{S_1^{AB}}.
\]
\end{enumerate}
\end{theorem}
\begin{proof}
\

\noindent 1. After Alice and Bob perform their quantum operations, we are left with the state:
\[
\lb T_{AA'} \otimes T_{BB'} \rb  \lb \ket{\psi}\ket{\phi}\bra{\psi}\bra{\phi}\rb,
\]
and the probability of winning is given by
\begin{align*}
tr\Big(\big(P_1 \otimes 1_{A'B'}\big)\big(T_{AA'} \otimes T_{BB'}\big)\big(\ket{\psi}\ket{\phi}\bra{\psi}\bra{\phi}\big)\Big)
\end{align*}
It follows from Stinespring's Theorem (see for instance \cite[Theorem 4.1]{Paulsen}) that we can write
\begin{align*}
T_{AA'}(\rho)=tr_{\mathcal K}
\Big(
 U_{AA'\mathcal K} (\rho \otimes \ket{0}\bra{0})  U_{AA'\mathcal K}^\dagger\Big),
\end{align*}
for a finite-dimensional Hilbert space $\mc K$ and a unitary $ U_{AA'\mc K}$ on $\mathcal H_A \otimes \mathcal H_{A'} \otimes \mc K$ (and similarly for $T_{BB'}$). It thus suffices to consider quantum operations of the form $T(x)=UxU^\dagger$, where $U$ is a unitary (by changing $\ket{\phi}$ to $\ket{\phi}\ket{00}$ and increasing the dimensions of $\mathcal H_{A'}$ and $\mathcal H_{B'}$).

The value of the game is then given by
\begin{align*}
\suvp tr
\Big(
 \big(
 P_1 \otimes 1_{A'B'}
 \big)
 \Big(
 U_{AA'} \otimes V_{BB'} \otimes 1_C
 )
\ket{\psi} \ket{\phi} \bra{\psi} \bra{\phi}
(
 U_{AA'}^\dagger \otimes V_{BB'}^\dagger \otimes 1_C
)\Big)
\Big),
\end{align*}
where $U_{AA'}$ and $V_{BB'}$ are unitaries on the indicated Hilbert spaces.

Recalling that $P_1=\ket{\gamma}\bra{\gamma}$, we can rewrite the above as
\begin{align*}
\suvp \Big\|
\lb
\bra{\gamma}\otimes 1_{A'B'}
\rb
\lb
U_{AA'} \otimes V_{BB'} \otimes 1_C
\rb
\ket{\psi}\ket{\phi}\Big\|^2_{\mathcal H_{A'}\otimes \mathcal H_{B'}}.
\end{align*}
Taking the supremum over $\ket{\phi} \in \mathcal H_{A'}\otimes \mathcal H_{B'}$, we can write this as an operator norm:
\begin{align}\label{min connection}
\omega^*(G)=& \suv \Big\|
\lb
\bra{\gamma}\otimes 1_{A'B'}
\rb
\lb
U_{AA'} \otimes V_{BB'} \otimes 1_C
\rb
\lb
\ket{\psi} \otimes 1_{A'B'}
\rb\Big\|^2_{B(\mathcal H_{A'}\otimes \mathcal H_{B'})}\\
\nonumber= & \suv \Big\|tr_{AB}
\Big(
\lb
U_{AA'} \otimes V_{BB'}
\rb
\big(M_{AB} \otimes 1_{A'B'}\big)\Big)\Big\|^2_{B(\mathcal H_{A'}\otimes \mathcal H_{B'})}.
\end{align}
On the other hand, as we explained in the previous section
\begin{align*}
\|M_{AB}\|_{S_1(\mathcal H_A)\otimes_{min} S_1(\mathcal H_B)}=\sup\big\|(T\otimes S)(M_{AB})\|_{B(\mathcal H_{A'})\otimes_{min} B(\mathcal H_{B'})},
\end{align*}
where the supremum is taken over all finite dimensional Hilbert spaces $\mathcal H_{A'}$ and $\mathcal H_{B'}$; and all completely contractions $T:S_1(\mathcal H_A)\rightarrow B(\mathcal H_{A'})$ and $S:S_1(\mathcal H_B)\rightarrow B(\mathcal H_{B'})$. Now, given such a $T$, the associated tensor $\hat{T}$ can be seen as an element in the unit ball of $B(\mathcal H_A)\otimes_{min} B(\mathcal H_{A'})=B(\mathcal H_A\otimes\mathcal H_{A'})$. Since we are interested in the extremal points we can restrict to unitaries $U_{AA'}$ in $B(\mathcal H_A\otimes\mathcal H_{A'})$. Using the same reasoning for $S$ we can restrict to unitaries $V_{BB'}$ in $B(\mathcal H_B\otimes\mathcal H_{B'})$. Therefore, the expression above can be written as
\begin{align*}
\|M_{AB}\|_{S_1(\mathcal H_A)\otimes_{min} S_1(\mathcal H_B)}=\suv \Big\|tr_{AB}
\Big(
\lb
U_{AA'} \otimes V_{BB'}
\rb
\big(M_{AB} \otimes 1_{A'B'}\big)\Big)\Big\|^2_{B(\mathcal H_{A'}\otimes \mathcal H_{B'})},
\end{align*}where the supremum is taken over all unitaries $U_{AA'}$ in $B(\mathcal H_A\otimes\mathcal H_{A'})$ and $V_{BB'}$ in $B(\mathcal H_B\otimes\mathcal H_{B'})$. This is exactly the same as the expression in (\ref{min connection}).

\noindent  2. Reasoning similarly to above, we obtain that the value of the game is given by
\begin{align*}
\suvc \Big\|
\lb
\bra{\gamma}\otimes 1_{A'}
\rb
\lb
V_{BA'}U_{AA'} \otimes 1_C
\rb
\ket{\psi}\ket{\phi}\Big\|^2_{\mathcal H_{A'}},
\end{align*}
where $U_{AA'}$ and $V_{BA'}$ are unitaries on the indicated Hilbert spaces. Rearranging as before we can write
\begin{align}
\omega_{qow}(G)=\sup_{U_{AA'},V_{BA'},}
\Big\|tr_{AB}\big((V_{BA'}U_{AA'})(M_{AB}\otimes 1_{A'})\big)\Big\|_{B(\mathcal H_{A'})}^2.
\end{align}
On the other hand, we already explained in Section \ref{sec operator spaces} that
\begin{align*}
\|M_{AB}\|_{S_1(\mathcal H_A)\otimes S_1(\mathcal H_B)}=\sup\|(T\odot S)(M)\|_{B(\mathcal H_{A'})}
\end{align*}where the $\sup$ is taken over all completely bounded norm one operators $T:S_1(\mathcal H_A)\rightarrow B(\mathcal H_{A'})$ and $S:S_1(\mathcal H_B)\rightarrow B(\mathcal H_{A'})$ and $(T\odot S)(x\otimes y)=T(x)S(y)$. In the same way as before, we can assume that $T$ and $S$ are operators associated to unitaries $U_{AA'}$ in $B(\mathcal H_A\otimes\mathcal H_{A'})$ and $V_{BA'}$ in $B(\mathcal H_B\otimes\mathcal H_{A'})$ respectively. Therefore, we obtain \begin{align*}
\|M_{AB}\|_{S_1(\mathcal H_A)\otimes S_1(\mathcal H_B)}=\sup_{U_{AA'},V_{BA'}}
\Big\|tr_{AB}\big((V_{BA'}U_{AA'})(M_{AB}\otimes 1_{A'})\big)\Big\|_{B(\mathcal H_{A'})},
\end{align*}as we wanted.

\noindent 3. Reasoning similarly to above, we replace the unitaries $U \in B(\mathcal H_A \otimes \mathcal H_{A'})$ and $V \in B(\mathcal H_B \otimes \mathcal H_{B'})$ by a single unitary $W \in B(\mathcal H_A \otimes \mathcal H_B \otimes \mathcal H_E)$, where $\mathcal H_E$ is the Hilbert space corresponding to the entangled state available to the player. Rearranging as before, we have that
\begin{align*}
V(G)=\sup_{W,\ket{\xi},\ket{\eta}} \Big|tr_{AB}\big(M_{AB}\bra{\xi}W_{ABE} \ket{\eta}\big)\Big|^2,
\end{align*}
where $\ket{\xi}, \ket{\eta}$ are unitary vectors in $\mathcal H_E$. Note, however, that this is exactly the same as taking
\begin{align*}
V(G)=\sup_{W} \Big|tr_{AB}\big(M_{AB}W_{AB}\big)\Big|^2,
\end{align*}where the supremum runs just on unitary operators in $B(\mathcal H_A \otimes \mathcal H_B)$.
Furthermore, the fact that $V(G)^{\frac{1}{2}}=\|M_{AB}\|_{S_1(\mathcal H_A\otimes \mathcal H_B)}$ follows trivially by the duality $(S_1^n)^*=M_n$ explained in the previous section.
\end{proof}
Note that we have $0\leq\omega^* (G)\leq \omega_{qow} (G)\leq V(G)\leq 1$ for every game $G$. However, Alice and Bob cannot define, in general, a perfect strategy. Indeed, the maximal value of the game $V(G)$ is in general strictly smaller than $1$. We finish this section by providing a characterization of those rank-one quantum games with maximal value equal to $1$.
\begin{lemma}\label{proj norm equal one}
Consider the states $\ket{\psi}=\sum_i \lambda_i \ket{i}_C\ket{\alpha_i}_{AB}$ and $\ket{\gamma}=\sum_i \mu_i \ket{i}_C \ket{\beta_i}_{AB}$.
The matrix $M_{AB}=tr_C \ket{\psi}\bra{\gamma}$ satisfies $\N{M_{AB}}_{S_1(\mathcal H_A) \otimes\mathcal H_B)}=1$ if and only if there is a unitary $U$ on $\mathcal H_A \otimes \mathcal H_B$ such that $\ket{\alpha_i}=U\ket{\beta_i}$ and $\lambda_i=\mu_i$ for all $i$.
\end{lemma}
We begin by fixing some notation.
\begin{align*}
\ket{\psi} =\sum_{i} \lambda_i \ket{i}\ket{\alpha_i}, & \text{   }\ket{\alpha_i}\in \mathcal H_A \otimes \mathcal H_B, & \N{\ket{\alpha_i}}_2=1, \quad &\lambda_i \geq 0, & \sum_i \lambda_i^2=1, \\
\ket{\gamma} =\sum_{i} \mu_i \ket{i}\ket{\beta_i}, & \text{   } \ket{\beta_i}\in \mathcal H_A \otimes \mathcal H_B, & \N{\ket{\beta_i}}_2=1, \quad  &\mu_i \geq 0 ,& \sum_i \mu_i^2=1,
\end{align*}
and write $M$ for the associated matrix $tr_C ( \ket{\psi}\bra{\gamma}) \in B(\mathcal H_A \otimes \mathcal H_B)$. We know that
$\N{M}_{S_1(\mathcal H_A \otimes \mathcal H_B)} \leq 1$ (see Proposition \ref{connection-projective}). We will use the following result from \cite{chan}:
\begin{theorem}
A subset $F$ is a proper closed face of the the unit ball of $S_1^n$ if and only if there exists a nonzero partial isometry $U$ such that
\[
F=\left\{ UP \st P \geq 0, \N{P}_1=1, \text{ and } \ker P \supseteq \ker U \right\}.
\]
\end{theorem}

\begin{proof}[Proof of Lemma \ref{proj norm equal one}]
Suppose that such a unitary $U$ exists and that $\lambda_i=\mu_i$ for all $i$. In this situation, we have that
\[
M=\sum_i \lambda_i^2 U\ket{\beta_i}\bra{\beta_i},
\]
and then $tr(U^*M)=\sum_i \lambda_i^2=1$. By the duality between $S_1^n$ and $M_n$ we conclude that $\|M\|_{S_1(\mathcal H_A\otimes \mathcal H_B)}\geq 1$. On the other hand, Proposition \ref{connection-projective} tells us that  $\|M\|_{S_1(\mathcal H_A\otimes \mathcal H_B)}\leq 1$, so we have equality in this expression.

For the converse, first note that we can assume that $\lambda_i$ and $\mu_i$ are non negative real numbers for every $i$. Let us assume that $\|M\|_{S_1(\mathcal H_A\otimes \mathcal H_B)}=1$. Then
\begin{align*}
1=\Big\|\sum_i \lambda_i \mu_i \ket{\alpha_i}\bra{\beta_i}\Big\|_{S_1(\mathcal H_A\otimes \mathcal H_B)} \leq \sum_i \lambda_i \mu_i \leq \left( \sum_i \lambda_i^2 \right)^{\frac{1}{2}} \left(\sum_i \mu_i^2 \right)^{\frac{1}{2}} \leq 1.
\end{align*}
Equality in the Cauchy-Schwartz inequality implies $\lambda_i = \mu_i$ for all $i$. We are thus considering the matrix $\sum_i \lambda_i^2 \ket{\alpha_i} \bra{\beta_i}$, a convex combination of elements in the unit sphere of $S_1^n$. If $\|M\|_{S_1(\mathcal H_A\otimes \mathcal H_B)}=1$, we must have that $\ket{\alpha_i} \bra{\beta_i}$ all lie in the same face of the unit ball of $S_1^n$. But in this case, there exist $P_i \geq 0$, $tr(P_i)=1$, and a partial isometry $U$ such that
\begin{align*}
\ket{\alpha_i} \bra{\beta_i}=UP_i, \quad \text{for all } i.
\end{align*}
We then have that $P_i= U^*\ket{\alpha_i} \bra{\beta_i}$. By positivity of $P_i$, we must have $\ket{\alpha_i}=U\ket{\beta_i}$. As $\mathcal H_A \otimes \mathcal H_B$ is finite dimensional, there exists a unitary $\tilde U$ such that $U|_{\supp U}=\tilde U|_{\supp U}$.
\end{proof}

\section{Three different models}\label{Section: three models}
A natural question is whether the three values $\omega^*$, $\omega_{qow} $ and $V$ are indeed different. It is very easy to see that the \emph{entangled value of the game with two way communication}, so the case when both players can send an unlimited amount of quantum information, already matches the value $V(G)$. A much more surprising result was given by Buhrman et al. (\cite[Theorem 4.1]{BCFGGOS}), who proved that in order to obtain the value $V(G)$ it is enough to consider the \emph{entangled value of the game with simultaneous mutual communication}. That is, when Alice and Bob can both send an unlimited amount of quantum information but their message cannot depend on the ones received.

Since some of the main results in the work deal with the problem of approximating some of these values up to a multiplicative constant factor (see Section \ref{Sec approx}), it is important to study whether the values $\omega^*$, $\omega_{qow} $ and $V$ are indeed different in this sense. That is, whether we can find certain games showing large gaps between the previous values. In this section we prove that this is indeed possible. Furthermore, we will provide optimal gaps (in the dimension of the games) between the entangled value and the entangled value with one-way communication of rank-one quantum games and also between the entangled value with one-way communication and the maximum value of these games. Let us start by defining the following two rank-one quantum games:
\begin{align*}
G_C=tr_C|\psi\rangle\langle\gamma|=\frac{1}{n}\sum_{i=1}^n|i\rangle\langle1|\otimes |1\rangle\langle i|,
\end{align*}where $|\psi\rangle=\frac{1}{\sqrt{n}}\sum_{i=1}^n|i1\rangle|i\rangle$ and $|\gamma\rangle=\frac{1}{\sqrt{n}}\sum_{i=1}^n|1i\rangle|i\rangle$; and
\begin{align*}
G_R=tr_C|\psi\rangle\langle\gamma|=\frac{1}{n}\sum_{i=1}^n|1\rangle\langle i|\otimes |i\rangle\langle 1|,
\end{align*}where $|\psi\rangle=\frac{1}{\sqrt{n}}\sum_{i=1}^n|1i\rangle|i\rangle$ and $|\gamma\rangle=\frac{1}{\sqrt{n}}\sum_{i=1}^n|i1\rangle|i\rangle$.

Let us briefly discuss the game $G_R$; Alice and Bob are asked to convert the state $\ket{\psi}$ to $\ket{\gamma}$. A player with access to both registers can map $\ket{1i}$ to $\ket{i1}$, so $V(G_R)=1$. If there are two space-like separated players, Alice initially shares no entanglement with the referee but is required to create a maximally entangled state shared with him; the ability to send quantum information to Bob does not help her perform this task. We would thus expect that $V(G_R) \gg \omega_{qow}(G_R)$.
Similar comments apply for the game $G_C$ where Bob and Alice interchange roles. However, if Alice is permitted to send him quantum information, he gains a clear advantage; in fact, $\omega_{qow}(G_C) =1$. It is then unsurprising that $\omega_{qow}(G_C) \gg \omega^*(G_C) $.
We can exactly calculate the values of these games and rigorously prove the above statements using simple arguments from operator space theory.

The main result of this section is as follows:
\begin{theorem}\label{gaps}
Let $G_C$ and $G_R$ be defined as above. Then,
\begin{enumerate}
\item[1.] $\frac{\omega_{qow}(G_C)}{\omega^*(G_C)}= n^2$,
\item[2.] $\frac{V(G_R)}{\omega_{qow}(G_R)}= n^2$.
\end{enumerate}
Moreover, such separations are optimal in the dimension of the games.
\end{theorem}
Let us start by explaining why these games are suitable for showing the aforementioned gaps exist as well as why these separations are optimal. In order to do this we will use the estimates given in Section \ref{sec operator spaces}. However, we will provide the precise computations in the corresponding proof of Theorem \ref{gaps} below.

According to Theorem \ref{main- connection}, in order to prove the optimality of the first part of Theorem \ref{gaps}, we must show that
\begin{align}\label{min-haagerup distance}
\N{id\otimes id: S_1^n \otimes_{min} S_1^n \to S_1^n \otimes_h S_1^n} = n.
 \end{align}
We note that, according to (\ref{CR in S_1}), $G_C$ can be regarded as the element $$\frac{1}{n}\sum_{i=1}^n|ii\rangle\in R_n\otimes C_n\subset S_1^n\otimes S_1^n.$$Since the minimal and the Haagerup tensor norms are injective, it suffices to study the corresponding norms on $R_n\otimes C_n$. By (\ref{min-operator norm}),
\begin{align*}
\big\|\frac{1}{n}\sum_{i=1}^n|ii\rangle\big\|_{R_n\otimes_{min}C_n}=\frac{1}{n}\big\|id:C_n\rightarrow C_n\big\|_{cb}=\frac{1}{n}.
\end{align*}On the other hand, according to (\ref{Row-Haagerup norm}) and (\ref{R fact CC}) we have
\begin{align*}
\big\|\frac{1}{n}\sum_{i=1}^n|ii\rangle\big\|_{R_n\otimes_{h}C_n}=\frac{1}{n}\Gamma_R(id:C_n\rightarrow C_n)=1.
\end{align*}This shows the lower bound in (\ref{min-haagerup distance}).

In order to show the converse inequality in (\ref{min-haagerup distance}) we decompose the identity as
\begin{align*}
id_{S_1^n \otimes_{min} S_1^n \to S_1^n \otimes_h S_1^n}=id_{C_{n^2} \otimes_{h} S_1^n \to S_1^n \otimes_h S_1^n}\circ id_{C_{n^2} \otimes_{min} S_1^n \to C_{n^2} \otimes_h S_1^n}\circ id_{S_1^n \otimes_{min} S_1^n \to C_{n^2} \otimes_{min} S_1^n}.
\end{align*}Therefore, it suffices to upper bound each of the maps. Now, according to (\ref{optimality gaps model}) we know that
\begin{align*}
\|id_{S_1^n \otimes_{min} S_1^n \to C_{n^2} \otimes_{min} S_1^n}\|=\|id_{C_{n^2} \otimes_h S_1^n \to S_1^n \otimes_h S_1^n}\|=\sqrt{n}.
\end{align*}On the other hand, according to (\ref{min-operator norm}), (\ref{Row-Haagerup norm}) and (\ref{R fact R}) we have
\begin{align*}
\|id_{C_{n^2} \otimes_{min} S_1^n \to C_{n^2} \otimes_h S_1^n}\|=1.
\end{align*} We conclude that
\begin{align*}
\|id_{S_1^n \otimes_{min} S_1^n \to S_1^n \otimes_h S_1^n}\|\leq \sqrt{n}\sqrt{n}=n.
\end{align*}
According to Theorem \ref{main- connection}, in order to prove the second part of Theorem \ref{gaps} we must show
\begin{align}\label{haagerup-proj distance}
\N{id\otimes id: S_1^n \otimes_{h} S_1^n \to S_1^{n^2}} = n.
\end{align}
We note that, according to (\ref{CR in S_1}), $G_R$ can be regarded as the element $$\frac{1}{n}\sum_{i=1}^n|ii\rangle\in C_n\otimes R_n\subset S_1^n\otimes S_1^n.$$Invoking again the injectivity of the Haagerup tensor product we can study the norm of the element in $C_n\otimes R_n$. Thus, according to (\ref{Row-Haagerup norm}) and (\ref{R fact CC}) we can conclude that
\begin{align*}
\big\|\frac{1}{n}\sum_{i=1}^n|ii\rangle\big\|_{C_n\otimes_{h}R_n}=\frac{1}{n}\Gamma_R(id:R_n\rightarrow R_n)=\frac{1}{n}.
\end{align*}On the other hand, one can easily deduce from Remark \ref{tensor product structures II} that
\begin{align*}
\|G_R\|_{S_1^{n^2}}=\big\|\frac{1}{n}\sum_{i=1}^n|ii\rangle\big\|_{S_1^n}=1.
\end{align*}Therefore, we obtain the lower bound in (\ref{haagerup-proj distance}).

In order to show the converse inequality in (\ref{haagerup-proj distance}) we decompose the identity as
\begin{align*}
id_{S_1^n \otimes_{h} S_1^n \to S_1^{n^2}}=id_{S_1^n \otimes_{h} S_1^n \to R_{n^2} \otimes_{h} S_1^n}\circ id_{R_{n^2} \otimes_{h} S_1^n \to S_1^{n^2}}.
\end{align*}According to (\ref{optimality gaps model}) we know that
\begin{align*}
\|id_{S_1^n \otimes_{h} S_1^n \to R_{n^2} \otimes_{h} S_1^n}\|=\|id_{R_{n^2} \otimes_{h} S_1^n \to S_1^{n^2}}\|=\sqrt{n}.
\end{align*}Therefore,
\begin{align*}
\|id_{S_1^n \otimes_{h} S_1^n \to S_1^{n^2}}\|\leq n.
\end{align*}
Let us now present a proof of $\frac{\omega_{qow}(G_C)}{\omega^*(G_C)}= n^2$ by detailing the previous estimates and following a more information theoretical point of view. This will serve as an illustration of how to translate some of the previous arguments to an information theoretical language.
\begin{proof}[Proof of (Theorem \ref{gaps}, 1.) and its optimality]
We first show that $\omega_{qow}(G_C)=1$ by giving a simple protocol between Alice and Bob with success probability equal to one. In our protocol $\mathcal H_{A'}=\C^n$ and Alice's initial state is initiated to be $|\varphi\rangle=|1\rangle$.
\begin{enumerate}
\item[1.] The initial state after Alice and Bob receive their particles is $|\psi\rangle|\varphi\rangle=\frac{1}{\sqrt{n}}\sum_{i=1}^n|i1\rangle|i\rangle|1\rangle$.
\item[2.] Then, Alice applies a flip on $A-A'$ to produce the state $\frac{1}{\sqrt{n}}\sum_{i=1}^n|11\rangle|i\rangle|i\rangle$ and sends her particle $A'$ to Bob.
\item[3.] Bob applies a flip on $A'$-$B$ to produce the state $\frac{1}{\sqrt{n}}\sum_{i=1}^n|1i\rangle|i\rangle|1\rangle$.
\item[4.] They return their particles to the referee.
\end{enumerate}
Since the final state is equal to $|\gamma\rangle$, once we trace out the auxiliary system, they will win with probability one.

Next we prove the upper bound $\omega^*(G_C)\leq \frac{1}{n^2}$. As we showed in the proof of Theorem \ref{main- connection} we have
\begin{align*}
\omega^*(G_C)^{\frac{1}{2}}=\sup\Big\{\Big\|tr_{AB}(U_{AA'}\otimes V_{BB'})G_C\Big\|_{B(\mathcal H_{A'}\otimes \mathcal H_{B'})}\Big\},
\end{align*}where the supremum runs over all unitaries $U_{AA'}\in B(\mathcal H_{A}\otimes \mathcal H_{A'})$ and $V_{BB'}\in B(\mathcal H_{B}\otimes \mathcal H_{B'})$. In our particular case, the previous expression is of the form \begin{align*}
\frac{1}{n}\sup\Big\{\Big\|\sum_{i=1}^n\langle1|U_{AA'}| i\rangle\otimes \langle i|V_{BB'}| 1\rangle \Big\|_{B(\mathcal H_{A'}\otimes \mathcal H_{B'})}\Big\}.
 \end{align*}Then, Cauchy-Schwartz inequality allows us to upper bound this expression by
\begin{align*}
\frac{1}{n}\sup\Big\{\Big\|\sum_{i=1}^n\langle1|U_{AA'}| i\rangle (\langle1|U_{AA'}| i\rangle)^* \Big\|^{\frac{1}{2}}_{B(\mathcal H_{A'})}\Big\}\sup\Big\{\Big\|\sum_{i=1}^n(\langle i|V_{BB'}| 1\rangle)^*\langle i|V_{BB'}| 1\rangle \Big\|^{\frac{1}{2}}_{B(\mathcal H_{B'})}\Big\}.
\end{align*}
Now, it is a very simple exercise to show that for every element (not necessarily unitary) $A\in B(\mathcal H\otimes \mathcal K)$ with $\|A\|\leq 1$ we have both inequalities
\begin{align*}
\Big\|\sum_{i=1}^n\langle 1|A| i\rangle (\langle 1|A| i\rangle)^* \Big\|^{\frac{1}{2}}_{B(\mathcal K)}\leq 1 \text{     }\text{   and   }\text{     }
\Big\|\sum_{i=1}^n(\langle i|A| 1\rangle)^*\langle i|A| 1\rangle \Big\|^{\frac{1}{2}}_{B(\mathcal K)}\leq 1.
\end{align*}This completes the proof of the estimate $\frac{\omega_{qow}(G_C)}{\omega^*(G_C)}\geq n^2$.

Finally, we will show that $\frac{\omega_{qow}(G)}{\omega^*(G)}\leq n^2$ for every rank-one quantum game $G$ of dimension $n$. Note that, being $G$ an element of $S_1^n\otimes S_1^n$, we can write it as $G=\sum_{i,j=1}^n|i\rangle\langle j|\otimes G_{i,j}$, where $G_{i,j}\in S_1^n$ for every $i,j$.
Now, according to Theorem \ref{main- connection} we have
\begin{align*}
\omega_{qow}(G)^\frac{1}{2}=\sup\Big\{\Big\|tr_{AB}(V_{BA'}\otimes U_{AA'})G\Big\|_{B(\mathcal H_{A'})},
\end{align*}where the supremum runs over all unitaries $U_{AA'}\in B(\C^n\otimes \mathcal H_{A'})$ and $V_{BA'}\in B(\C^n\otimes \mathcal H_{A'})$. Note that for every unitary $U_{AA'}\in B(\C^n\otimes \mathcal H_{A'})$, we can write $U=\sum_{i,j=1}^n|i\rangle\langle j|\otimes U_{i,j}$ with $U_{i,j}\in B(\mathcal H_{A'})$ and such that
\begin{align*}
\Big\|\sum_{i,j=1}^n(U_{i,j})^*U_{i,j}\Big\|_{B(\mathcal H_{A'})}^\frac{1}{2}\leq \Big(\sum_{i=1}^n\Big\|\sum_{j=1}^n(U_{i,j})^*U_{i,j}\Big\|_{B(\mathcal H_{A'})}\Big)^\frac{1}{2}\leq \sqrt{n}.
\end{align*}Therefore, we can use Cauchy-Schwartz inequality to upper bound $\omega_{qow}(G)^\frac{1}{2}$ by
\begin{align*}
\sup\Big\{\Big\|\sum_{i,j=1}^n U_{i,j}\big(tr_{B}(V_{BA'}G_{i,j})\big)\Big\|_{B(\mathcal H_{A'})}\leq \sqrt{n}\Big\|\sum_{i,j=1}^n\big(tr_{B}(V_{BA'}G_{i,j})\big)\big(tr_{B}(V_{BA'}G_{i,j})\big)^*\Big\|^\frac{1}{2}.
\end{align*}Furthermore, by Cauchy-Schwartz inequality, we can upper bound the previous quantity by
\begin{align*}
\sup\Big\{\Big\|tr_{AB}(U_{AA''}\otimes V_{BA'})G\Big\|_{B(\mathcal H_{A''}\otimes \mathcal H_{A'})}\Big\},
\end{align*}where the supremum runs over all unitaries $V_{BA'}\in B(\C^n\otimes \mathcal H_{A'})$ and all matrices $U_{AA''}=\sum_{i,j=1}^n|i\rangle\langle j|\otimes \tilde{U}_{i,j}\in B(\C^n\otimes \mathcal H_{A''})$ such that
\begin{align*}
\Big\|\sum_{i,j=1}^n(\tilde{U}_{i,j})^*\tilde{U}_{i,j}\Big\|_{B(H_{A''})}^\frac{1}{2}\leq 1.
\end{align*}Hence, using that $\big\|\sum_{i,j=1}^n|i\rangle\langle j|\otimes A_{i,j}\big\|_{M_n(B(\mathcal H))}\leq \sqrt{n}\|\sum_{i,j=1}^nA_{i,j}^*A_{i,j}\|_{B(\mathcal H)}^\frac{1}{n}$, we conclude, by Theorem \ref{main- connection}, that
\begin{align*}
\omega_{qow}(G)^\frac{1}{2}\leq n\omega^*(G)^\frac{1}{2}.
\end{align*}
\end{proof}
We encourage the reader to do the corresponding computations for the second part of Theorem \ref{gaps}.
\section{Computing and approximation the different values of rank-one games}\label{Sec approx}

The main result of this section is
\begin{theorem}\label{main-theorem}
There exists an efficient (i.e. polynomial in the dimension of the Hilbert space associated to the questions and answers) algorithm to approximate the entangled value of a rank-one game $G$ within a factor of 4.
\end{theorem}
To prove that, the operator space Grothendieck inequality (Theorem \ref{Grothendieck II}) tells that
$$\|M_{AB}\|^2_\mu\le \omega^*(G)\le 4\|M_{AB}\|^2_\mu\; ,$$ where
$$\|M\|_\mu=\sup \Big\{\langle M,A \rangle: \max\big\{\|A\|_{B(\He_A)\otimes_hB(\He_B)},\|A^t\|_{B(\He_A)\otimes_hBB}\big\}\leq 1\Big\}.$$
We will show that $\|\cdot \|_\mu$ can be computed efficiently. By a classical result of Yudin and Nemirovskii relating membership and optimization problems in convex bodies (see \cite[Section 4]{GLS}), it is enough to show that one can efficiently compute $\|\cdot\|_{B(\He_A)\otimes_hB(\He_B)}$. Also, without loss of generality, we can assume $\He_A=\He_B$ and call it simply $\He$. The next step is to use that
\begin{theorem}\label{HaagerupMnCB}
 The map
\begin{align*}
\Delta: B(\He)\otimes_hB(\He)\rightarrow CB(B(\He),B(\He)),
\end{align*}defined by $\Delta(u)(A)=\sum_ix_iA y_i$ is an isometric isomorphism, where $u=\sum_ix_i\otimes y_i$ is any decomposition with $x_i,y_i\in\mathcal{B}(\He)$.
\end{theorem}
A proof of this Theorem, originally due to Haagerup, can be found in  \cite[Theorem 5.12]{Pisier}\label{lem1} or \cite[Theorem 11]{JKP}. The proof of Theorem \ref{main-theorem} can now be finished by using the following result due to Watrous:
\begin{theorem}[\cite{Wat}, \cite{JKP}]\label{SDP Watrous}
The completely bounded norm of any linear map $T:B(\He)\rightarrow B(\He)$ can be computed efficiently in $\dim{\He}$. Moreover, such a norm can be expressed by a semidefinite program.
\end{theorem}
Since the entangled value of the game with one-way communication is given by the Haagerup tensor norm on $S_1^A\otimes S_1^B$, assuming again without loss of generality $\He_A=\He_B$ and using the self duality of the Haagerup tensor norm (\ref{duality haagerup}) one obtains from Theorem \ref{SDP Watrous} that

\begin{prop}
For any rank one game $G$, the value $\omega_{qow}(G)$ can be computed efficiently. The same happens (trivially) for the maximal value $V(G)$.
\end{prop}

Indeed, one can easily find an SDP for $\omega_{qow}(G)$.  One option for that is to rely on the reduction (see Section \ref{sec parallel repetition} below) to a four-round single prover quantum interactive proof, for which an SDP is known \cite{Gutoski}. In the light of the recent paper \cite{RegevVidick}, one can adopt a more direct approach:  using the self duality of the Haagerup tensor norm, together with its definition (\ref{eq:Haagerup}) and the definitions of the row and column operator spaces (\ref{CR structures}) one obtains directly that $\|M_{AB}\|_{S_1^A\otimes_hS_1^B}$ is given by
\begin{align*}
&\max{\langle M_{AB}, u\rangle}\quad  {\rm s.t.}\\
 &  \begin{array}{l} u= \sum_i A_i\otimes B_i\\
A_i\in B(\He_A), B_i \in B(\He_B) \; \; \forall i\\
 \sum_i A_i A_i^*\le id_{\He_A}\\
 \sum_i B_i^*B_i\le id_{\He_B}
 \end{array}
\end{align*}
It is not difficult to see, e.g. reasoning as in \cite{RegevVidick}, that this optimization program is an SDP. Indeed, this was the idea (albeit exchanging (\ref{eq:Haagerup}) with the characterization of $\|\cdot \|_\mu$ given in \cite[Theorem 5.18, (ii)]{Pisier}), that Regev and Vidick used  in \cite{RegevVidick} to improve our Theorem \ref{main-theorem} by showing an SDP for $\|\cdot \|_\mu$.
\section{Parallel repetition of rank-one quantum games}\label{sec parallel repetition}
In this section we will study the behavior of the different values of a rank-one quantum game with respect to the perfect parallel repetition theorem. In particular, the main theorem of this section shows that a perfect parallel repetition for the entangled value of rank-one quantum games dramatically fails (See Theorem \ref{Parallel entangled value}). First, note that if we have a rank-one quantum game $G$ with initial state $\ket{\psi}$ and final measurement $P=\ket{\gamma}\bra{\gamma}$ and we consider $\ket{\psi},\ket{\gamma} \in \mathcal H_A \otimes \mathcal H_B \otimes \mathcal H_C$ and $M_{AB}=tr_C \ket{\psi}\bra{\gamma}\in S_1^n\otimes S_1^n$, then, the parallel repetition of this game $k$ times is given by the rank-one quantum game $G^k$ defined by an initial state $\ket{\psi}^{\otimes_k}$ and final measurement $P^{\otimes_k}$. It is easy to see that this corresponds to the element $M_{AB}^k=tr_C \ket{\psi}^{\otimes_k}\bra{\gamma}^{\otimes_k}\in S_1^{n^k}\otimes S_1^{n^k}$ ob
 tained by the $k^{th}$ tensor product of $M_{AB}$.
Let us first start by showing that a perfect parallel repetition theorem holds for the maximal value and the entangled value with one-way communication of rank-one quantum games.
\subsection{A perfect parallel repetition for $V(G)$ and $\omega_{qow}(G)$}
\begin{theorem}\label{perfect parallel Haagerup}
\begin{enumerate}
\

\item[1.] The maximal value $V$ verifies a perfect parallel repetition theorem on rank-one quantum games.
\item[2.] The entangled value with one-way communication $\omega_{qow}$ verifies a perfect parallel repetition theorem on rank-one quantum games.
\end{enumerate}
\end{theorem}
As we will see the first part of the previous theorem is very easy, so the main part of Theorem \ref{perfect parallel Haagerup} is the statement on $\omega_{qow}$. In fact, we were kindly informed by Thomas Vidick and John Watrous that the entangled value with one-way communication of quantum games (not necessarily rank-one) can be seen as a special case of a single-prover $QIP(4)$. Indeed, given one of such a quantum games defined via an initial state $|\psi\rangle_{ABC}$ and final measurement $(V_0,V_1)$, we can convert it into a single-prover $QIP(4)$ as follows:
\begin{enumerate}
\item[1.] The verifier prepares $|\psi\rangle_{ABC}$, and sends the $A$ register to the single prover $P$.
\item[2.] $P$ sends back an answer register $A'$.
\item[3.] The verifier sets $A'$ aside, and sends the $B$ part of $|\psi\rangle_{ABC}$ to the prover.
\item[4.] The prover sends back an answer $B'$. The verifier measures $(A',B',C)$ using $(V_0,V_1)$.
\end{enumerate}
An elementary analysis shows that the value of this game coincides with the entangled value with one-way communication of the initial quantum game. With this reduction at hand, one could deduce a perfect parallel repetition theorem from \cite[Theorem 4.9]{Gutoski}. However, we will present here a very simple proof by using standard results from operator space theory. We believe our proof illustrates the power of the operator spaces techniques and that the connection to the Haagerup tensor norm might have further applications in some other related problems. In order to present our simple proof, let us introduce the minimal tensor product in its ``complete form'' (that is, as an operator space). Suppose that $V$ and $W$ are operator spaces and that $u =\sum_{k=1}^l a_k\otimes v_k \otimes w_k\in M_n(V\otimes W)$, then the \emph{minimal tensor product} norm of $u$ is
\begin{align}\label{minimal operator space}
\|u\|_{M_n(V\otimes_{min}W)}=\sup \Big\{\Big \|\sum_{k=1}^l a_k \otimes T(v_k) \otimes S(w_k)\Big \|_{B(\C^n\otimes \mathcal H_V \otimes \mathcal H_W)}
\Big\},
\end{align}
where the supremum is taken over all finite dimensional Hilbert spaces $\mathcal H_V$ and $\mathcal H_W$ and complete contractions $T:V \rightarrow B(\mathcal H_V)$ and $S:W\rightarrow B(\mathcal H_W)$. Equation (\ref{minimal operator space}) defines an operator space structure on the minimal tensor product and it is not difficult to see that Equation (\ref{metric mapping property}) becomes
\begin{align*}
\big\|T\otimes S:V\otimes_{min} W\rightarrow V'\otimes_{min} W'\big\|_{cb}=\|T\|_{cb}\|S\|_{cb}.
\end{align*}The following proposition, which will be the key point for the proof of Theorem \ref{perfect parallel Haagerup}, is a direct consequence of the previous estimate.
\begin{prop}\label{multiplicaivity}
Let $T:\C^{n^2}\rightarrow \C^l$ be a linear map. Then,
\begin{align*}
\big\|T^{\otimes_k}:M_{n^k}\rightarrow R_{l^k}\big\|_{cb}=\|T:M_n\rightarrow R_l\|^k_{cb},
\end{align*}where we have just used that $M_{n^k}=\otimes^k_{min} M_n$ and $R_{l^k}=\otimes^k_{min} R_l$ as operator spaces\footnote{Note that the first assertion is trivial and the second one is the same as Remark \ref{tensor product structures I}.}; and the same happens if we replace $R_l$ by $C_l$. Furthermore, by duality we have the same multiplicativity property for maps $T:R_l\rightarrow S_1^n$ and $T:C_l\rightarrow S_1^n$.
\end{prop}
We defined the operator space structure of the minimal tensor product so that we could state Proposition \ref{multiplicaivity} as a known result. However, it is a simple exercise to prove Proposition \ref{multiplicaivity} using the definitions given in Section \ref{sec operator spaces}. Since the proof of Theorem \ref{perfect parallel Haagerup} also implicitly uses the equivalence between the two definitions for the Haagerup tensor product explained in Section \ref{sec operator spaces}, which has not been proved in this paper, we will not care about using Proposition \ref{multiplicaivity} as a known result.
\begin{proof}[Proof of Theorem \ref{perfect parallel Haagerup}]
To show the first assertion, let us consider a game $G\in S_1(\mathcal H_A\otimes \mathcal H_B)$. Then, we can write $G=\sum_{i=1}^N\alpha_i|f_i\rangle\langle g_i|$, where $(f_i)_{i=1}^N$ and $(g_i)_{i=1}^N$ are orthonormal bases of $\mathcal H_A \otimes \mathcal H_B$ and $(\alpha_i)_{i=1}^N$ is a sequence of positive real numbers verifying $\sum_{i=1}^N\alpha_i= V(G)$. Therefore, we know that  $G^{\otimes_k}=\sum_{i_1,\cdots, i_k=1}^N\alpha_{i_1}\cdots \alpha_{i_k}|f_{i_1}\rangle\langle g_{i_1}|\otimes\cdots \otimes |f_{i_k}\rangle\langle g_{i_k}|$. However, we can realize this element as $G^{\otimes_k}=\sum_{i_1,\cdots, i_k=1}^N\alpha_{i_1}\cdots \alpha_{i_k}|f_{i_1}\otimes \cdots \otimes f_{i_k}\rangle\langle g_{i_1}\otimes \cdots \otimes g_{i_k}|$, so
\begin{align*}
V(G^{\otimes_k})=\sum_{i_1,\cdots, i_k=1}^N\alpha_{i_1}\cdots \alpha_{i_k}=V(G)^k.
\end{align*}
In order to prove the second part of the theorem, we must show that $\omega_{qow}(G^{\otimes k}) = \omega_{qow}(G)^k$. Since the inequality $\omega(G^{\otimes k}) \geq \omega(G)^k$ is trivial we have to focus on the other one. We recall that Theorem \ref{main- connection} tells us that $\omega_{qow}(G)^{1/2}=\|G\|_{S_1^n \otimes_h S_1^n}$. On the other hand, if the game $G$ corresponds to the linear map $\tilde G:M_n \to S_1^n$, then the game $G^{\otimes k}$ will correspond to the map $\tilde G^{\otimes k}: M_{n^k} \to S_1^{n^k}$.
By Equation (\ref{Row-Haagerup norm}), we know that
\begin{align*}
\|G\|_{S_1^n \otimes_h S_1^n}=\Gamma_R(\tilde G) \text{      }\text{   and  }
\|G^{\otimes k}\|_{S_1^{n^k} \otimes_h S_1^{n^k}}=\Gamma_R(\tilde G^{\otimes k}).
\end{align*}The first inequality tells us that for every $\epsilon> 0$ there exist two linear maps $a:M_n\rightarrow R_l$, $b:R_l\rightarrow S_1^n$ for a certain natural number $l$ such that $\tilde G=b\circ a$ and $\|G\|_{S_1^n \otimes_h S_1^n}\leq \|a\|_{cb}\|b\|_{cb}+\epsilon$. On the other hand, by Proposition \ref{multiplicaivity} we know that the maps $a^{\otimes k}:M_{n^k}\rightarrow R_{l^k}$ and $b^{\otimes k}:R_{l^k}\rightarrow S_1^{n^k}$ verifies $\|a^{\otimes k}\|_{cb}\|b^{\otimes k}\|_{cb}= \|a\|^k_{cb}\|b\|^k_{cb}$. Since $\tilde G^{\otimes k}=b^{\otimes k}\circ a^{\otimes k}$ and this happens for every $\epsilon> 0$, we conclude that $\|G^{\otimes k}\|_{S_1^{n^k} \otimes_h S_1^{n^k}}\leq \|G\|^k_{S_1^n \otimes_h S_1^n}$ as we wanted.
\end{proof}
\begin{remark}
After establishing the isometric identification $CB(M_n,M_n)=M_n\otimes_h M_n$ (Theorem \ref{HaagerupMnCB}), we could invoke the multiplicativity of the diamond norm (\cite{KW}), so the completely bounded norm, to state a similar result for the elements in $M_n\otimes_h M_n$. On the other hand, the entangled value with one-way communication of a rank one quantum game is defined by the corresponding norm on $S_1^n\otimes_h S_1^n=(M_n\otimes_h M_n)^*$. One could then try to prove the multiplicativity of the norm $S_1^n\otimes_h S_1^n$ (so perfect parallel repetition for the entangled value with one-way communication of a rank-one quantum game) by connecting these two ideas. However, we did not find any direct argument for this.
\end{remark}
\subsection{No perfect parallel repetition theorem for the entangled value of rank-one quantum games}
Let us consider the following rank-one quantum game:
\begin{align*}
G_{C+R}=\frac{1}{2}(G_C+G_R)=tr_C(|\psi\rangle\langle\gamma|)=\frac{1}{2n}\sum_{i=1}^n\Big(|i\rangle\langle1|\otimes |1\rangle\langle i|+ |1\rangle\langle i|\otimes |i\rangle\langle 1|\Big),
\end{align*}where $|\psi\rangle=\frac{1}{\sqrt{2n}}\sum_{i=1}^n\Big(|i1\rangle|i1\rangle+|1i\rangle|i2\rangle\Big)$ and
$|\gamma\rangle=\frac{1}{\sqrt{2n}}\sum_{i=1}^n\Big(|1i\rangle|i1\rangle+|i1\rangle|i2\rangle\Big)$.

The main result of this section states as follows.
\begin{theorem}\label{Parallel entangled value}
The following equalities hold:
\begin{align*}
\omega^*(G_{C+R})=\frac{1}{n^2}, \text{     }{     }\omega^*(G_{C+R}^2)=\frac{1}{4n^2}(1+\frac{1}{n})^2.
\end{align*}In particular,
\begin{align*}
\frac{\omega^*(G_{C+R}^2)}{(\omega^*(G_{C+R}))^2}\succeq n^2,
\end{align*}where $\succeq$ inequality up to a universal constant independent on $n$.
\end{theorem}
Let us first explain our motivation to consider this game and we will show the precise calculations later.
Since our game is just a convex combination of the elements considered in Section \ref{gaps}, we have
\begin{align*}
\|G_{C+R}\|_{S_1^n\otimes_{min}S_1^n}\leq \frac{1}{2}\Big(\|G_{C}\|_{S_1^n\otimes_{min}S_1^n}+\|G_{R}\|_{S_1^n\otimes_{min}S_1^n}\Big)=\frac{1}{n}.
 \end{align*}According to Theorem \ref{main- connection} this implies that $\omega^*(G_{C+R})\leq \frac{1}{n^2}$.
In fact, it is very easy to see that we have an equality. On the other hand, the tensor product of this game with itself expands in four different terms:
\begin{align*}
G_{C+R}^2= G_C^2+ G_C\otimes G_R+G_R\otimes G_C + G_R^2.
\end{align*}We need to compute the minimal norm of this tensor as an element in $S_1^{n^2}\otimes_{min} S_1^{n^2}$.

The game $G_C$ can be understood as the element $\frac{1}{n}\sum_{i=1}^n|ii\rangle\in R_n\otimes C_n\subset S_1^n\otimes S_1^n$. Then, in order to realize the tensor product in the suitable space we will obtain
\begin{align*}
G_C^2=\frac{1}{n^2}\sum_{i,j=1}^n|ij\rangle\otimes |ij\rangle\in (R_n\otimes R_n)\otimes (C_n\otimes C_n)\subset S_1^{n^2}\otimes S_1^{n^2}.
\end{align*}
According to Remark \ref{tensor product structures II} we have
\begin{align*}
\|G_C^2\|_{S_1^{n^2}\otimes_{min} S_1^{n^2}}=\frac{1}{n^2}\|\sum_{i,j=1}^n|ij\rangle\otimes |ij\rangle\|_{R_{n^2}\otimes_{min}C_{n^2}}=
\frac{1}{n^2}\|id:C_{n^2}\rightarrow C_{n^2}\|_{cb}=\frac{1}{n^2}.
\end{align*}The same argument applies to $G_R^2$.

On the other hand, if we look at the term $G_C\otimes G_R$ we see that we must compute the minimal norm of the element \begin{align*}
G_C\otimes G_R=\frac{1}{n^2}\sum_{i,j=1}^n|ij\rangle\otimes |ij\rangle\in (R_n\otimes C_n)\otimes (C_n\otimes R_n)\subset S_1^{n^2}\otimes S_1^{n^2}.
\end{align*}According again to Remark \ref{tensor product structures II} we have
\begin{align*}
\|G_C\otimes G_R\|_{S_1^{n^2}\otimes S_1^{n^2}}=\frac{1}{n^2}\|\sum_{i,j=1}^n|ij\rangle\otimes |ij\rangle\|_{S_1^n\otimes_{min}S_1^n}\\=
\frac{1}{n^2}\|id:M_n\rightarrow S_1^n\|_{cb}= \frac{1}{n^2}\|id:M_n\rightarrow S_1^n\|=\frac{1}{n},
\end{align*}and analogously for $G_R\otimes G_C$.

Of course, we must check that there are not cancelations between the four previous terms and this will be done below by showing the exact calculations. On the other hand, the fact that the last lower bound is attained for the norm (rather than the completely bounded norm) tells us that the corresponding strategy will not need any extra shared entanglement between Alice and Bob. Let us now present a detailed proof of Theorem \ref{Parallel entangled value}.
\begin{proof}
We have already seen the upper bound $\omega^*(G_{C+R})\leq \frac{1}{n^2}$. On the other hand, this value for the game can be obtained by following the trivial strategy. That is, Alice and Bob do not do anything on their particles! Indeed, in this case the winning probability is given by
\begin{align*}
\langle \gamma| \psi\rangle^2=\frac{1}{n^2}.
\end{align*}
Let us now show the estimate $\omega^*(G_{C+R}^2)\geq \big(\frac{1}{2n}(1+\frac{1}{n})\big)^2$. In fact, as we mentioned before, we should expect to obtain a strategy with no extra entanglement shared by Alice and Bob. Such a strategy is very simple once we look at the states which define the game $G_{C+R}^2$. Indeed, note that such a game is defined by
\begin{align*}
|\psi\rangle^{\otimes_2}=\frac{1}{2n}\sum_{i,j=1}^n\Big(|ij\rangle_A|11\rangle_B|ij11\rangle_C+
|i1\rangle_A|1j\rangle_B|ij12\rangle_C\\+ |1j\rangle_A|i1\rangle_B|ij21\rangle_C+ |11\rangle_A|ij\rangle_B|ij22\rangle_C\Big),
\end{align*}and
\begin{align*}
|\gamma\rangle^{\otimes_2}=\frac{1}{2n}\sum_{i,j=1}^n\Big(|11\rangle_A|ij\rangle_B|ij11\rangle_C+
|1j\rangle_A|i1\rangle_B|ij12\rangle_C\\+ |i1\rangle_A|1j\rangle_B|ij21\rangle_C+ |ij\rangle_A|11\rangle_B|ij22\rangle_C\Big).
\end{align*}
Therefore, it is trivial to check that if we consider the strategy in which Alice and Bob apply a flip on their corresponding parts of $|\psi\rangle^{\otimes_2}$ (with no extra shared entanglement) we obtain the state
\begin{align*}
|\psi'\rangle=\frac{1}{2n}\sum_{i,j=1}^n\Big(|ji\rangle_A|11\rangle_B|ij11\rangle_C+
|1i\rangle_A|j1\rangle_B|ij12\rangle_C\\+ |j1\rangle_A|1i\rangle_B|ij21\rangle_C+ |11\rangle_A|ji\rangle_B|ij22\rangle_C\Big),
\end{align*}which verifies
\begin{align*}
\langle\gamma^{\otimes_2}|\psi'\rangle^2=\Big(\frac{1}{4n^2}(2+2n)\Big)^2=\frac{1}{4n^2}(1+\frac{1}{n})^2.
\end{align*}
In order to finish the proof of Theorem \ref{Parallel entangled value}, we will show that the previous lower bound is optimal. Since
\begin{align*}
G_{C+R}^2=\frac{1}{4n^2}\sum_{i,j=}^n\Big(|ij\rangle\langle 11|\otimes |11\rangle\langle ij|+ |i1\rangle\langle 1j|\otimes |1j\rangle\langle i1|\\+ |1j\rangle\langle i1|\otimes |i1\rangle\langle 1j|+|11\rangle\langle ij|\otimes |ij\rangle\langle 11|\Big),
\end{align*}according to Theorem \ref{main- connection} we have that $\omega^*(G_{C+R}^2)^\frac{1}{2}$ is equal to
\begin{align*}
\frac{1}{4n^2}\sup\Big\{\Big\|\sum_{i,j=1}^n\Big(\langle 11|U_{AA'}|ij\rangle\otimes \langle ij|V_{BB'}|11\rangle + \langle 1j|U_{AA'}|i1\rangle\otimes \langle i1|V_{BB'}|1j\rangle\\+ \langle i1|U_{AA'}|1j\rangle\otimes \langle 1j|V_{BB'}|i1\rangle+ \langle jj|U_{AA'}|11\rangle\otimes \langle 11|V_{BB'}|ij\rangle\Big)\Big\|_{B(\mathcal H_{A'}\otimes \mathcal H_{B'})}\Big\};
\end{align*}which is upper bounded by
\begin{align}\label{four temrs}
\frac{1}{4n^2}\Big[\sup\Big\{\Big\|\sum_{i,j=1}^n\langle 11|U_{AA'}|ij\rangle\otimes \langle ij|V_{BB'}|11\rangle\Big\|_{B(\mathcal H_{A'}\otimes \mathcal H_{B'})}\Big\} \\\nonumber+\sup\Big\{\Big\|\sum_{i,j=1}^n\langle 1j|U_{AA'}|i1\rangle\otimes \langle i1|V_{BB'}|1j\rangle\Big\|_{B(\mathcal H_{A'}\otimes \mathcal H_{B'})}\Big\}\\ \nonumber+\sup\Big\{\Big\|\sum_{i,j=1}^n\langle i1|U_{AA'}|1j\rangle\otimes \langle 1j|V_{BB'}|i1\rangle\Big\|_{B(\mathcal H_{A'}\otimes \mathcal H_{B'})}\Big\}\\\nonumber+
\sup\Big\{\Big\|\sum_{i,j=1}^n\langle jj|U_{AA'}|11\rangle\otimes \langle 11|V_{BB'}|ij\rangle\Big\|_{B(\mathcal H_{A'}\otimes \mathcal H_{B'})}\Big\}\Big].
\end{align}Here, the supremum is taken over unitaries $U_{AA'}\in B(\C^{n^2}\otimes \mathcal H_{A'})$ and $U_{BB'}\in B(\C^{n^2}\otimes \mathcal H_{B'})$ in all cases.
Now, note that for every such unitaries Cauchy-Schwarts inequality tells that $\Big\|\sum_{i,j=1}^n\langle 11|U_{AA'}|ij\rangle\otimes \langle ij|V_{BB'}|11\rangle\Big\|_{B(\mathcal H_{A'}\otimes \mathcal H_{B'})}$ is upper bounded by
\begin{align*}
\Big\|\sum_{i,j=1}^n\langle 11|U_{AA'}|ij\rangle(\langle 11|U_{AA'}|ij\rangle)^*\Big\|^\frac{1}{2}_{B(\mathcal H_{A'})}\Big\|\sum_{i,j=1}^n(\langle ij|V_{BB'}|11\rangle)^*\langle ij|V_{BB'}|11\rangle\Big\|^\frac{1}{2}_{B(\mathcal H_{B'})}\leq 1,
\end{align*}and the same argument works for the last term in (\ref{four temrs}).
On the other hand, for every such unitaries the quantity $\Big\|\sum_{i,j=1}^n\langle 1j|U_{AA'}|i1\rangle\otimes \langle i1|V_{BB'}|1j\rangle\Big\|_{B(H_{A'}\otimes H_{B'})}$ is upper bounded by
\begin{align*}
\sum_{i=1}^n\Big\|\sum_{j=1}^n\langle 1j|U_{AA'}|i1\rangle\otimes \langle i1|V_{BB'}|1j\rangle\Big\|_{B(\mathcal H_{A'}\otimes \mathcal H_{B'})}\leq n,
\end{align*}and the same argument works for the third term in the sum above. Therefore, we have that
\begin{align*}
\omega^*(G_{C+R}^2)^\frac{1}{2}\leq \frac{1}{4n^2}(2+ 2n)=\frac{1}{2n}(1+ \frac{1}{n}),
\end{align*}which concludes the proof.
\end{proof}
Finally, the following easy lemma completes the information about the game $G_{C+R}$.
\begin{lemma}
\

\begin{enumerate}
\item[1.] $V(G_{C+R})=V(G_{C+R}^2)=1$.
\item[2.] $\omega_{qow}(G_{C+R})=\frac{1}{4}\big(1+ \frac{1}{n}\big)^2$; $\omega_{qow}(G_{C+R}^2)=\frac{1}{16}\big(1+ \frac{1}{n}\big)^4$.
\end{enumerate}
\end{lemma}
\begin{proof}
To see the first part just note that we can obtain the state $|\gamma\rangle$ from the state $|\psi\rangle$ by simply flipping Alice's and Bob's part of the state. Then, we have $V(G_{C+R})=1$. On the other hand, this immediately implies that $V(G_{C+R}^2)=1$.

To show the estimate $\omega_{qow}(G_{C+R})=\frac{1}{4}\big(1+ \frac{1}{n}\big)^2$ we first note that
\begin{align*}
\|G_{C+R}\|_{S_1^n\otimes_h S_1^n}\leq \frac{1}{2}\Big(\|G_{C}\|_{S_1^n\otimes_h S_1^n}+ \|G_{R}\|_{S_1^n\otimes_h S_1^n}\Big)\leq \frac{1}{2}(1+\frac{1}{n}),
\end{align*}where in the last inequality we have used the estimates given in Section \ref{gaps}. According to Theorem \ref{main- connection} we conclude that $\omega_{qow}(G_{C+R})\leq \frac{1}{4}(1+\frac{1}{n})^2$.

Let us now prove that $\omega_{qow}(G_{C+R})$ attains the upper bound shown above. In order to do this they will perform the following strategy. In our protocol $H_{A'}=\C^n$ and Alice's initial state is initiated in $|\varphi\rangle=|1\rangle$.
\begin{enumerate}
\item[1.] The initial state after Alice and Bob receive their particles is $ |\psi\rangle|\varphi\rangle=\frac{1}{\sqrt{2n}}\sum_{i=1}^n\big(|i1\rangle|i1\rangle|1\rangle+ |1i\rangle|i2\rangle|1\rangle\big)$.
\item[2.] Then, Alice applies a flip on $A-A'$ to produce the state $\frac{1}{\sqrt{2n}}\sum_{i=1}^n\big(|11\rangle|i1\rangle|i\rangle+ |1i\rangle|i2\rangle|1\rangle\big)$ and sends her particle $A'$ to Bob.
\item[3.] Bob applies a flip on $A'-B$ to produce the state $$|\varphi'\rangle=\frac{1}{\sqrt{2n}}\sum_{i=1}^n\big(|1i\rangle|i1\rangle|1\rangle+ |11\rangle|i2\rangle|i\rangle\big).$$
\item[4.] They return their particles to the referee.
\end{enumerate}The winning probability can be easily checked to be $\frac{1}{4}\big(1+\frac{1}{n}\big)^2$.

Finally, in order to conclude $\omega_{qow}(G_{C+R}^2)=\frac{1}{16}\big(1+ \frac{1}{n}\big)^4$ we invoke Theorem \ref{perfect parallel Haagerup}.
\end{proof}

\section{Two families of games from operator space theory}\label{Section:SchurOHn}

\subsection{Schur Games}

We have seen that Rank-One Quantum Games correspond to elements in $S_1(\h_A) \otimes S_1(\h_B)$. They then correspond to maps from $B(\h_A)$ to $S_1(\h_B)$ via the identification
\[
\left(\sum_{i=1}^n A_i \otimes B_i \right) (C)=\sum_{i=1}^n tr(A_i C^{tr})B_i.
\]
(The transpose is due to the fact that we use the scalar pairing (\ref{Scalar Pairing})). We now define Schur games in terms of the corresponding maps from $B(\h_A)$ to $S_1(\h_B)$.
\begin{definition}
Let $G$ be a quantum entangled game with initial state $\ket{\psi}$ and final measurement $\ket{\gamma}\bra{\gamma}$ for $\ket{\psi},\ket{\gamma} \in \h_A \otimes \h_B \otimes \h_C$. Let $M_{AB}=tr_C \ket{\psi}\bra{\gamma}$. We will say that $G$ is a \emph{Schur game} if the associated map $\hat{M}_{AB}:B(\h_A)\rightarrow S_1(\h_B)$ to $M_{AB}$ is a Schur multiplier. That is, there exists an element $\phi \in B(\h_A)$ verifying that $$\hat{M}_{AB}(X)=\phi* X$$ for every $X\in B(\h_A)$. Here $*$ denotes the Schur (or Hadamard) product defined by $\bra{i} \phi * X \ket{j} = \phi_{ij}X_{ij}$.
\end{definition}
\begin{remark}\label{Shur-dimension}
Note that the dimensions of Alice and Bob's Hilbert spaces are the same for Schur games. In what follows, we take this dimension to be $n$.  Schur games are those that correspond to the ``diagonal'' elements in $S_1(\h_A) \otimes S_1(\h_B)$, those of the form \mbox{$\sum_{i,j=1}^n \phi_{ij} \ket{i}\bra{j} \otimes \ket{i}\bra{j}$}.
\end{remark}
\begin{remark}\label{RemarkSchurBimodule}
Let $D_{\vec{x}}$ denote the diagonal matrix whose diagonal entries are given by $\vec{x} \in \ell_n^\infty$. Then it is straightforward to check that a map $T:B(\h_A) \to S_1(\h_B)$ is a Schur multiplier if and only if $T(D_{\vec{x}} A D_{\vec{y}}) = D_{\vec{x}} T(A) D_{\vec{y}}$ for all $\vec{x}, \vec{y} \in \ell_n^\infty$. If $T$ satisfies this bimodule condition, then the  matrix $\phi$ is given by $\phi_{ij} = \langle T( \ket{i}\bra{j}) , \ket{i}\bra{j} \rangle$.
\end{remark}
\subsubsection{Characterization of Schur Games}
In this section we will give a very simple characterization of Schur games.
\begin{lemma}\label{SchurCharacterization}
A rank-one quantum game is a Schur game if and only if it has a representation $M=tr_C \ket{\psi}\bra{\gamma}$ where $\ket{\psi}$ and $\ket{\gamma}$ are of the form
\begin{align*}
\ket{\psi}&=\sum_{i,t} \alpha_{it}\ket{i}_A\ket{i}_B\ket{t}_C, \\
\ket{\gamma}&=\sum_{i,t} \beta_{it} \ket{i}_A\ket{i}_B\ket{t}_C,
\end{align*}
with $\N{\alpha}_2 \leq 1$, $\N{\beta}_2\leq 1$.
\end{lemma}
\begin{proof}
Assume that the rank-one quantum game is a Schur game, i.e., after fixing suitable orthonormal bases, the element $\hat M_{AB}$ in $CB(M_n,S_1^n)$ corresponding to the game is of the form $(x_{ij})\mapsto (t_{ij}x_{ij})$. We now examine the relationship between the element $M_{AB}\in CB(M_n,S_1^n)$ and the states $\ket{\psi}$ and $\ket{\gamma}$.

By taking a partial trace over $\h_C$, we obtain
\[
tr_C{\ket{\psi} \bra{\gamma}}= \sum_{r,s,t,u,v} \alpha_{rst}\overline{\beta_{uvt}} \ket{r}\bra{u}_A \otimes \ket{s}\bra{v}_B.
\]
Using the identification $S_1^n \otimes_{min} S_1^n\simeq CB(M_n,S_1^n)$ and using the scalar pairing to identify $M_n^*$ and $S_1^n$, we apply this to $\ket{i}\bra{j}\in M_n$ to get
\begin{align*}
\Big(tr_C{\ket{\psi} \bra{\phi}}\Big)\Big(\ket{i}\bra{j}\Big)&=\sum_{s,v}\left( \sum_{r,u,t} \alpha_{rst}\overline{\beta_{uvt}} \langle i\mid r\rangle \langle u\mid j\rangle \right) \ket{s}\bra{v}\notag \\
&= \sum_{s,t,v} \alpha_{ist}\overline{\beta_{jvt}} \ket{s}\bra{v}.
\end{align*}
For this to be a Schur multiplier, the above sum must equal $t_{ij}\ket{i}\bra{j}$. Replacing $\alpha_{ist}$ by $\delta_{st}\alpha_{ist}$ and $\beta_{jvt}$ by $\delta_{jv}\beta_{jvt}$ will not alter $tr_C{\ket{\psi} \bra{\gamma}}$ and thus we can express the states as
\begin{align*}
\ket{\psi}&=\sum_{i,t} \alpha_{iit}\ket{i}_A\ket{i}_B\ket{t}_C, \\
\ket{\gamma}&=\sum_{j,t} \beta_{jjt} \ket{j}_A\ket{j}_B\ket{t}_C,
\end{align*}
with $\sum_{i,t} |\alpha_{iit}|^2 \leq 1$, $\sum_{j,t} |\beta_{jjt}|^2 \leq 1$.
The reverse implication is trivial.
\end{proof}
\subsubsection{Equivalence of Entangled and One-Way Communication Values}\label{SchurEquivalenceSubsection}
In this section we show that the entangled value $\omega^*(G)$ and the entangled value of the game with one-way communication $\omega_{qow}(G)$ are equivalent for the particular games that we are considering. More precisely, we will show
\begin{theorem}
Let $G$ be a Schur game. Then the entangled value of $G$ and the entangled value of the game with one-way communication of $G$ are equivalent up to a multiplicative constant:
\[
\omega^*(G) \leq \omega_{qow}(G) \leq 4 \, \omega^*(G).
\]
\end{theorem}
\begin{remark}
It is important to notice here that, even when the entangled value with no communication and with one-way communication are ``very close'', the amount of resources needed to attain such values can be very different in both situations. An extremal example is given by the game of Leung, Toner, and Watrous \cite{LTW}. In this case the initial state is $|\psi\>=\frac{1}{\sqrt{2}}|000\>_{ABC}+\frac{1}{2}(|11\>+|22\>)_{AB}|1\>_C$ whereas the final state is $|\gamma\>=\frac{1}{\sqrt{2}}(|000\>+|111\>)_{ABC} $. By Lemma \ref{SchurCharacterization}, this is a Schur multiplier game. They show in \cite{LTW} that one needs infinite entanglement in order to play the optimal strategy for the entangled value $\omega^*(G)$ which is in this case equal to $1$. It is however straightforward to see that with an auxiliary system $A'$ of dimension $2$, if Alice applies to $AA'$ a unitary with $|00\>\mapsto |00\>$, $\ket{10}\mapsto \ket{11}$, and $\ket{20}\mapsto\ket{10}$ and later Bob applies on $BA'$ a u
 nitary with  $|00\>
 \mapsto |00\>$, and $\frac{1}{\sqrt{2}}(|11\>+|20\>)\mapsto \ket{10}$, they can obtain $|\gamma\>$ starting from $|\psi\>$ with the communication of just one qubit.
\end{remark}
The following argument depends on the Grothendieck inequality for operator spaces and is taken from Section 4 of \cite{PS} (using improved constants from \cite{HM}). We include it here for the convenience of the reader.

Suppose that $G$ is a Schur game. It corresponds to an element in $S_1(\h_A) \otimes S_1(\h_B)$ of the form
\[
\sum_{ij} \phi_{ij} \ket{i}\bra{j} \otimes \ket{i}\bra{j}.
\]
Let us denote by $u:B(\h_A) \rightarrow S_1(\h_B)$ the associated linear map to $G$. Note that, by definition,
\[
\langle u(a),b\rangle =  \sum_{ij} \phi_{ij} a_{ij} b_{ij},
\]for ever $a\in B(\h_A)$ and $b\in B(\h_B)$.

We claim that there exists a matrix $( \psi_{ij})$ such that
\begin{align}\label{s1sup}
\N{(\psi_{ij})}_1\leq 2 \omega^*(G)^{1/2} \text{ and } \vert \phi_{ij} \vert \leq \vert \psi_{ij}\vert,
\end{align}
for all $i$ and $j$. Here, we denote $\N{(\psi_{ij})}_1=\|\sum_{i,j}\psi_{ij}| i\rangle\langle j|\|$.

Let $(\psi_{ij})$ be such a matrix. Then, it can be written in the form $\sum_k \lambda_k \ket{x^k}\bra{y^k}$ where $\lambda_k \geq 0$, $\N{\ket{x_k}}_2=\N{\ket{y_k}}_2=1$ and $\sum_k \lambda_k \leq 2 \omega^*(G)^{1/2}$. By setting $X_i=\left( \sum_k \lambda_k |\langle x^k|i\rangle|^2 \right)^{1/2}$ and $Y_j = \left( \sum_k \lambda_k |\langle y^k|j\rangle|^2 \right)^{1/2}$, we obtain the following: there exist vectors $(X_i)_{i=1}^n$ and $(Y_j)_{j=1}^n$ such that
\begin{align} \label{l2sup}
| \phi_{ij} | \leq X_i Y_j \text{    }\text{ and }\text{    } \N{(X_i)}_2 \cdot \N{(Y_j)}_2\leq 2\omega^*(G)^{1/2}.
\end{align}
Clearly, (\ref{l2sup}) also implies the existence of a matrix $(\psi_{ij})$ satisfying (\ref{s1sup}).

Let $(X_i)$ and $(Y_j)$ be vectors satisfying (\ref{l2sup}); let $f=\sum_i |X_i|^2 \ket{i}\bra{i}$ and $g=\sum_j |Y_j|^2 \ket{j}\bra{j}$. Then we have
\begin{align*}
|\langle u(a), b \rangle | & = | \sum_{ij} \phi_{ij} a_{ij} b_{ij}| \leq \sum_{ij} \vert X_i \vert  \vert a_{ij} \vert \vert b_{ij} \vert \vert Y_j \vert \\
 & \leq \left(
 \sum_{ij}  \vert X_i \vert^2 \vert a_{ij} \vert^2
 \right)^{1/2}
 \left(\sum_{ij}  \vert Y_j \vert^2 \vert b_{ij} \vert^2
 \right)^{1/2} \\
 & = \left( f(aa^*) g(b^*b) \right)^{1/2}
\end{align*}
We then have that, for all $a_i \in B(\h_A)$, $b_i \in B(\h_B)$, $i \in \{1, \dots, m\}$,
\begin{align*}
\left\vert \sum_i \langle u(a_i), b_i \rangle \right\vert & \leq \sum_i f(a_i a_i^*)^{1/2} g(b_i^*b_i)^{1/2} \\
& \leq \left(
            \sum_i f (a_i a_i^*)
            \right)^{1/2}
            \left(
            \sum_i g (b_i^*b_i)
            \right)^{1/2}\\
 & =  \left(
             f ( \sum _ia_i a_i^*)
            \right)^{1/2}
            \left(
             g (\sum_i b_i^*b_i)
            \right)^{1/2}\\
& \leq 2\omega^*(G)^{1/2} \N{\sum_i a_i a_i^*}^{1/2} \N{\sum_i b_i^*b_i}^{1/2}.
\end{align*}
This shows that when $u$ is considered as a linear functional on $B(\h_A) \otimes_h B(\h_B)$, it has norm less than or equal to $2 \omega^*(G)^{1/2}$. But by (\ref{duality haagerup}) we have that $(B(\h_A) \otimes_h B(\h_B))^*=S_1(\h_A) \otimes_h S_1(\h_B)$ and then by Theorem \ref{main- connection}(2), we have that
\[
\omega_{qow}(G)^{1/2} = \N{u}_{S_1(\h_A) \otimes_h S_1(\h_B)}^{1/2} \leq 2 \, \omega^*(G)^{1/2}.
\]
In order to conclude the proof we must prove our claim (\ref{s1sup}). According to Theorem \ref{Grothendieck} (see Equation (\ref{u1Grothendieck})), there is a decomposition of linear maps $u=u_1 + u_2$ and states $f_1,f_2 \in S_1(\h_A)$ and $g_1, g_2 \in S_1(\h_B)$ such that
\begin{align*}
|\langle u_1(a),b\rangle| \leq \omega^*(G)^{1/2} f_1(aa^*)^{1/2} g_1(b^*b)^{1/2}, 
|\langle u_2(a),b\rangle| \leq \omega^*(G)^{1/2} f_2(a^*a)^{1/2} g_2(bb^*)^{1/2}, 
\end{align*}
for all $a \in B(\h_A)$ and $b \in B(\h_B)$.

The linear maps $u_1$ and $u_2$ are not necessarily associated with Schur multipliers. Using an averaging argument, we replace them by linear maps $\tilde u_1$ and $\tilde u_2$ that do correspond to Schur multipliers. Let $G$ be the group of all diagonal unitary matrices on $\ell_2^n$ equipped with its normalized Haar measure $\mathrm{m}$. We define linear maps $\tilde u_i$ by means of
\[
\langle\tilde u_i(a),b\rangle = \int_{G \times G} \langle u_i (x a y), x^{-1} b y^{-1}\rangle \, \mathrm{dm}(x) \, \mathrm{dm}(y).
\]
Using (\ref{u1Grothendieck}) and the Cauchy-Schwartz inequality, we obtain
\[
|\langle\tilde u_1(a),b\rangle| \leq \omega^*(G)^{1/2}  \left( \int f_1(xaa^*x^{-1}) \, \mathrm{dm}(x) \right)^{1/2} \left( \int g_1(y b^* b y^{-1}) \, \mathrm{dm}(y) \right)^{1/2}.
\]
A similar inequality holds for $|\langle\tilde u_2(a),b\rangle|$.

By the translation invariance of the Haar measure, we have that $\langle\tilde u_i(x a y), x^{-1} b y^{-1}\rangle=\langle\tilde u_i (a),b\rangle$. It follows from Remark \ref{RemarkSchurBimodule} that the linear maps $\tilde u_i$ correspond to Schur multipliers. In fact, the matrices $(\phi^i_{kl})= (\tilde u_i (\ket{k}\bra{l}, \ket{k} \bra{l}) ) $ satisfy the relations
\begin{align*}
\langle\tilde u_i(a),b\rangle = \sum_{ij} \phi^i_{kl} a_{kl} b_{kl}.
\end{align*}
As $u = \tilde u = \tilde u_1 + \tilde u_2$, we have that $\phi = \phi^1 + \phi^2$.

The states $\tilde f^1( \cdot ) = \int_G f_1(x \cdot x^{-1}) \, \mathrm{dm}(x)$ and $\tilde g^1( \cdot) = \int_G g_1(y \cdot y^{-1}) \, \mathrm{dm}(y)$ are diagonal states.
We thus have that
\[
\left| \sum_{ij} \phi^1_{ij} a_{ij} b_{ij} \right| \leq \omega^*(G)^{1/2} \left( \sum_{ij} \tilde f^1_{ii} | a_{ij} |^2 \sum_{ij} \tilde g^1_{jj}|b_{ij}|^2  \right)^{1/2},
\]
and that
\[
\left| \phi^1_{ij} \right| \leq \omega^*(G)^{1/2} |\tilde f^1_{ii}|^{1/2} | \tilde g^1_{jj}|^{1/2},
\]
where $\sum_i  |\tilde f^1_{ii}|=1 = \sum_j | \tilde g^1_{jj}|$. Similarly, we can obtain the bound
\[
\left| \phi^2_{ij} \right| \leq \omega^*(G)^{1/2} |\tilde f^2_{jj}|^{1/2} | \tilde g^2_{ii}|^{1/2}.
\]
If we let $\psi$ be the matrix $\psi_{ij}= \omega^*(G)^{1/2}\left( |\tilde f^1_{ii}|^{1/2} | \tilde g^1_{jj}|^{1/2} + |\tilde f^2_{jj}|^{1/2} | \tilde g^2_{ii}|^{1/2}\right)$, it then follows that
\[
| \phi_{ij} | = | \phi^1_{ij} + \phi^2_{ij}| \leq |\psi_{ij}|,
\]
where $\N{\psi}_1$ satisfies
\[
\N{\psi}_1 \leq \omega^*(G)^{1/2} \left(\N{|(\tilde f^1_{ii}|^{1/2} | \tilde g^1_{jj}|^{1/2})_{i,j}}_1 +\N{(|\tilde f^2_{jj}|^{1/2} | \tilde g^2_{ii}|^{1/2})_{i,j}}_1 \right)\leq 2 \omega^*(G)^{1/2}.
\]
This completes the proof of (\ref{s1sup}).
\begin{remark} \label{sgsup}
Let $G$ be a Schur game with matrix $(\phi_{ij})$ and let $S(G)$ denote the following quantity
\[
S(G) = \inf \left\{
                   \N{\psi}_1 \, : \, \vert\phi_{ij}\vert \leq \vert \psi_{ij} \vert \text{ for all } i,j
                   \right\}.
\]
The proof of the preceding theorem also shows that the following chain of inequalities holds:
\[
\omega^*(G)^{1/2} \leq \omega_{qow}(G)^{1/2} \leq S(G) \leq 2 \, \omega^*(G)^{1/2}.
\]
\end{remark}
\subsubsection{Separation of the maximal value and the entangled value}
A Schur game $G$ corresponds to an element of the form
\[
M_\phi = \displaystyle \sum_{i,j=1}^n \phi_{i,j} \ket{i}\bra{j} \otimes \ket{i}\bra{j} \text{ in } S_1(\h_A) \otimes S_1(\h_B).
\]
Despite this simplicity, we will now show that the class of Schur games is rich enough to contain games whose maximal value and entangled value differ by an arbitrarily large multiplicative factor. We begin by recalling that $V(G)^{1/2}= \N{M_\phi}_{S_1^{n^2}}$.  This norm is easy to calculate for Schur games; it is the trace-norm of the matrix $(\phi_{i,j})$:
\[
\N{\sum_{i,j=1}^n \phi_{i,j} \ket{i}\bra{j} \otimes \ket{i}\bra{j} }_{S_1^{n^2}} = \N{( \phi_{i,j})  }_1.
\]
We seek to compare this to the entangled value $\omega^*(G)^{1/2}$. By Remark \ref{sgsup}, for Schur games we have
\[
\omega^*(G)^{1/2} \simeq \inf \{ \N{\psi}_1 \, : \,  |\phi_{ij}| \leq |\psi_{ij}|,  \forall i,j  \},
\]
where $\simeq$ denotes equivalence up to a universal constant.

Remark 5.6 from \cite{Oikhberg} then suggests how we can separate the values $\omega^*(G)$ and $V(G)$. Using \cite{Simon}, we can find matrices $\phi$ and $\psi$ for which $|\phi_{ij}| \leq |\psi_{ij}|$ for all $i,j$, but $\N{\phi}_1 \gg \N{\psi}_1$. If $G$ is the game corresponding to the Schur multiplier $M_\phi$, then $V(G)= \N{\phi}_1^2$ but $\omega^*(G)$ is bounded by the far smaller $\N{\psi}_1^2$.

Consider the following matrices:
\begin{align*}
A =
\left(
\begin{array}{cc}
1 & 1 \\
1 & -1
\end{array}
\right)
\quad \text{and} \quad
B =
\left(
\begin{array}{cc}
1 & 1 \\
1 & 1
\end{array}
\right).
\end{align*}
We have that $\N{A}_1=2 \sqrt{2}$ and $\N{B}_1=2$ and $|A_{ij}| \leq |B_{ij}|$.

It immediately follows that
\begin{align*}
\left| (A^{\otimes n})_{ij} \right|  \leq  (B^{\otimes n})_{ij} \text{    }\text{   for every  }\text{     } i,j,\\
\N{ A^{\otimes n}}_1 = \left( 2 \sqrt{2} \right)^n= 2^{n/2} \N{B^{\otimes n}}_1.
\end{align*}
%
%
We now renormalize to obtain
\[
\phi_n =  2^{- \frac{3}{2}n} A^{\otimes n}.
\]

For each $n \in \mathbb N$, consider the rank-one quantum game $A_n$ with the following initial and final states:
\begin{align*}
\ket{\psi_n} & = 2^{-n}
\left(
    \left(
\ket{00}+\ket{11} \
   \right)_{AB} \ket{0}_C
 +  \left(
\ket{00} - \ket{11}
   \right)_{AB} \ket{1}_C
\right)^{\otimes n}   \\
\ket{\gamma_n} &= 2^{-n/2}
\left(
\ket{00}_{AB} \ket{0}_C + \ket{11}_{AB}\ket{1}_C
\right)^{\otimes n}
\end{align*}
The game $A_n$ corresponds to the Schur multiplier with matrix $\phi_n$.
We know by the previous comments that $V(G)=1$. In fact, this is very easy since a player with access to the registers of both Alice and Bob merely has to map each $\frac{1}{\sqrt{2}} \left( \ket{00} + \ket{11} \right)$ to $\ket{00}$ and each $\frac{1}{\sqrt{2}} \left(\ket{00} - \ket{11}\right)$ to $\ket{11}$. On the other hand, since $G$ is a Schur game we can conclude from Remark \ref{sgsup} that
\[
\omega^*(G)^{1/2} \leq \inf \{ \N{\psi}_1 \, : \,  |\phi_{ij}| \leq |\psi_{ij}|, \text{    } \forall i,j  \}.
\]
If we now apply this result with $\psi_n = 2^{- \frac{3}{2}n} B^{\otimes n}$, it follows that $\frac{1}{2^n} \geq \omega^*(A_n)$. We thus have that
\[
V(A_n) =1 \gg \frac{1}{2^n} \geq \omega^*(A_n).
\]
\subsection{$OH_n$-Games}
We conclude this work by introducing a second family of rank-one quantum games, motivated by the \emph{Operator Hilbert space}, $OH_n$. This is an operator space structure on $\C^n$ defined by
\[
\N{\sum_{i=1}^n x_i \otimes |i\rangle}_{M_N(\C^n)}= \N{\sum_{i=1}^n x_i\otimes \overline{x_i}}_{M_{N^2}}^{1/2},
\]
where $\sum_{i=1}^n x_i \otimes |i\rangle \in M_N(\C^n)$.

If $V$ is a vector space, then $\overline{V}$ is the same vector space but with the conjugate multiplication by a complex scalar: $\lambda \cdot \overline x = \overline{\bar  \lambda \cdot x}$, for $\lambda \in \mathbb C$, $x \in V$. Here $\overline x$ is the element in $\overline V$ corresponding to the element $x \in V$. If $V \subseteq B(H)$ is an operator space, then its \emph{conjugate} operator space structure is given by the corresponding embedding $\overline{V} \subseteq \overline{B(H)} = B(\overline H)$. Also, given a linear map $T:V\rightarrow W$ between operator spaces we can define $\overline{T}:\overline{V}\rightarrow \overline{W}$ as $\overline{T}(\overline{x})=\overline{T(x)}$ for every $\overline{x}\in \overline{V}$. It can be shown that $$OH_n^* = \overline{OH_n}$$ is a completely isometric identification.
\begin{definition}
Let $G$ be a quantum entangled game with initial state $\ket{\psi}$ and final measurement $\ket{\gamma}\bra{\gamma}$ for $\ket{\psi},\ket{\gamma} \in \h_A \otimes \h_B \otimes \h_C$. Let $M=tr_C \ket{\psi}\bra{\gamma}\in S_1(\mathcal H_A)\otimes S_1(\mathcal H_B)$ and let us denote by $\widehat{M}:B(\mathcal H_A)\rightarrow S_1(\mathcal H_B)$ the corresponding linear map. We will say that $G$ is a \emph{$OH_n$-game} if $\text{rank}(\widehat{M})=n$ and $tr\big(\widehat{M}(x)x^*\big)\geq 0$ for every $x\in B(\mathcal H_A)$.
\end{definition}
It is shown in \cite{Pisier} that in  this case $\widehat{M}=\overline{V^*}V$ for a certain linear map $V:B(\mathcal H_A)\rightarrow OH_n$ such that $\|\widehat{M}\|_{cb}=\|V\|_{cb}^2$.

Theorem \ref{Parallel entangled value} motivates the study of those games for which perfect parallel repetition is true or, at least, for which the quotient appearing in such a theorem cannot be large. This is the main reason to consider the $OH_n$ games. As we will see this family of games verifies some good multiplicativity properties. More precisely,
\begin{theorem}\label{violation OH-gameI}
Let $G$ be a $OH_n$ game. Then,
\begin{align*}
\frac{\omega^*(G^k)}{\omega^*(G)^k}\leq C^k (1+\ln n)^{2k},
\end{align*}where $n$ is the rank of $G$, $k$ is any natural number and $C$ is a universal constant independent of $n$ and $k$.
\end{theorem}
On the other hand, one can show that the order $(1+\ln n)^{2k}$ in Theorem \ref{violation OH-gameI} is essentially optimal for these games. Specifically,
\begin{theorem}\label{lowerbound OH_n}
For every natural number $n$, there exists an $OH_n$-game $G$ such that
\begin{align*}
\frac{\omega^*(G^k)}{\omega^*(G)^k}\geq C_1C_2^k \frac{(1+\ln n)^{2k}}{(1+k\ln n)^2}
\end{align*}for every natural number $k$, where $C_1$ and $C_2$ are universal constants independent of $n$ and $k$.
\end{theorem}
Therefore, even when Theorem \ref{violation OH-gameI} tells us that perfect parallel repetition theorem is not far from being true for $OH_n$-games, it still fails for these kinds of games.

Unfortunately, constructions involving the $OH_n$ operator space are usually tough owing to the lack of nice embeddings $OH_n\hookrightarrow B(\mathcal H)$ and $OH_n\hookrightarrow S_1(\mathcal H)$. In particular, the proof of Theorems \ref{violation OH-gameI} and \ref{lowerbound OH_n} require the use of highly nontrivial techniques from operators spaces. Developing these proofs for a non-specialist reader would require a significant extension of the length of this work. Therefore we prefer to add brief proofs of these results in the appendix below. The interested reader will be able to provide the finer details of the proofs after carefully checking the corresponding references.
\section*{Acknowledgments}
During the process of preparing this work we learned that Oded Regev and Thomas Vidick came up with similar connections while  considering different kinds of games \cite{RegevVidick}. We must thank both authors for a very generous and complete reading of our preliminary draft. We also thank the referees for their detailed and helpful reports.
\section{Appendix: Proofs of Theorem \ref{violation OH-gameI} and Theorem \ref{lowerbound OH_n}}
\subsection{Proof of Theorem \ref{violation OH-gameI}}
Let us assume for simplicity that $\mathcal H_A=\mathcal H_B=\C^N$ so that $B(\mathcal H_A)=M_N$ and $S_1(\mathcal H_A)=S_1^N$. According to Theorem \ref{main- connection}, it suffices to show that
\begin{align*}
\frac{\|\otimes^k \widehat{M}:M_{N^k}\rightarrow S_1^{N^k}\|_{cb}}{\|\widehat{M}:M_N\rightarrow S_1^N\|_{cb}^k}\leq C^\frac{k}{2} (1+\ln n)^k.
\end{align*}
Our observation is that we have
\begin{align*}
\|\otimes^k V:M_{N^k}\rightarrow OH_{n^k}\|_{cb}\leq \pi_2^o(\otimes^k V:M_{N^k}\rightarrow OH_{n^k})=\pi_2^o(V:M_N\rightarrow OH_n)^k.
\end{align*}Here, $\pi_2^o$ denotes the completely $2$-summing norm (see \cite{Pisier2}, \cite{Junge-Groth}).
Therefore,
\begin{align*}
\|\otimes^k \widehat{M}:M_{N^k}\rightarrow S_1^{N^k}\|_{cb}\leq \pi_2^o(V:M_N\rightarrow OH_n)^{2k}.
\end{align*}The main point is that $\pi_2^o(V)\leq c_0\sqrt{1+\ln n} \|V\|_{cb}$ for every map $V:M_N\rightarrow OH_n$, where $c_0$ is a universal constant independent of $n$ (see \cite[Equation (2.3)]{Junge-Groth}) and, hence, we obtain
\begin{align*}
\|\otimes^k \widehat{M}:M_{N^k}\rightarrow S_1^{N^k}\|_{cb}\leq c_0^{2k}(1+\ln n)^k\|\widehat{M}:M_N\rightarrow S_1^N\|_{cb}^k
\end{align*}as we wanted.
\subsection{Proof of Theorem \ref{lowerbound OH_n}}
The construction relies on the following result proved in \cite{Junge-Groth}:
\begin{theorem}\label{thm-Junge-Groth}
There exist universal constants $C_0> 0$ and $C_0'>0$ such that
\begin{enumerate}
\item[1.] For every $n$ we can find a natural number $N$, a complete contraction $u_n:OH_n\rightarrow M_N$ and a linear map $w_n:M_N\rightarrow OH_n$ such that $\|w_n\|_{cb}\leq C_0\sqrt{\frac{n}{1+\ln n}}$ and verifying $w_n\circ u_n=id:\ell_2^n\rightarrow \ell_2^n$.
\item[2.] The previous factorization is optimal. That is, for every $n$ and for every maps $u:OH_n\rightarrow B(\ell_2)$, $w:B(\ell_2)\rightarrow OH_n$ verifying $w\circ u=id:\ell_2^n\rightarrow \ell_2^n$ we have $\|u\|_{cb}\|w\|_{cb}\geq C_0'\sqrt{\frac{n}{1+\ln n}}$.
\end{enumerate}
\end{theorem}
The above result is Corollary 4.11 in \cite{Junge-Groth} (see also \cite{Junge-Xu}). The inequality missing in the statement of that result follows immediately from combining the argument of Corollary 4.11 (i.e., $\gamma_\infty(id_{OH_n}) \pi_1^o(id_{OH_n})=n$) with the estimates for $\pi_1^o(id_{OH_n})$ given by Corollary 4.8 and Proposition 4.9.

We want to show that for some $\widehat{M}:M_N\rightarrow S_1^N$ we have
\begin{align*}
\frac{\|\otimes^k \widehat{M}:M_{N^k}\rightarrow S_1^{N^k}\|_{cb}}{\|\widehat{M}:M_N\rightarrow S_1^N\|_{cb}^k}\geq \sqrt{C_1}C_2^\frac{k}{2} \frac{(1+\ln n)^k}{1+k\ln n}.
\end{align*}
Fixing $n$ and $k$ we consider the map $w_n:M_N\rightarrow OH_n$ of Theorem \ref{thm-Junge-Groth}. The first property above tells us that
\begin{align*}
\|w_n:M_N\rightarrow OH_n\|_{cb}^k\leq C_0^k\Big(\sqrt{\frac{n}{1+\ln n}}\Big)^k.
 \end{align*}On the other hand, since $\|\otimes^ku_n:OH_{n^k}\rightarrow M_{N^k}\|_{cb}\leq 1$, the second part of Theorem \ref{thm-Junge-Groth} applied to $OH_{n^k}$ tells us that
 \begin{align*}
 \|\otimes^k w_n:M_{N^k}\rightarrow OH_{n^k}\|\geq C_0'\sqrt{\frac{n^k}{1+k\ln n}}.
 \end{align*}
Let us consider now the $OH_n$-game defined by the map $$\widehat{M}:\overline{w_n^*}\circ w_n:M_N\rightarrow \overline{OH_n}\simeq OH_n^*\rightarrow S_1^N.$$
Note that $$\otimes^k \widehat{M}=(\otimes^k\overline{w_n^*})\circ (\otimes^k w_n):M_{N^k}\rightarrow \overline{OH_{n^k}}\simeq OH_{n^k}^*\rightarrow S_1^{N^k}.$$Then, as we mentioned before, $$\|\widehat{M}:M_N\rightarrow S_1^N\|_{cb}=\|w_n:M_N\rightarrow OH_n\|_{cb}^2$$ and $$\|\otimes^k\widehat{M}:M_{N^k}\rightarrow S_1^{N^k}\|_{cb}=\|\otimes^kw_n:M_{N^k} \rightarrow OH_{n^k}\|_{cb}^2.$$Therefore,
$$\frac{\|\otimes^k\widehat{M}:M_{N^k}\rightarrow S_1^{N^k}\|_{cb}}{\|\widehat{M}:M_N\rightarrow S_1^N\|_{cb}^k}\geq \frac{(C_0')^2\frac{n^k}{1+k\ln n}}{C_0^{2k}\Big(\frac{n}{1+\ln n}\Big)^k}= (C_0')^2 \left(\frac{1}{C_0}\right)^{2k}\frac{(1+\ln n)^k}{1+k\ln n}.$$
\begin{remark}
The preceding construction can be altered to provide also the maximal value of the game involved. According to \cite{Junge-Groth} (see also \cite{HM1}), for every $n$ we can find an $N$ and a complete embedding $j_n:OH_n\hookrightarrow S_1^N$ such that $\|j_n\|\|j_n^{-1}\|\leq C$ for certain universal constant $C$ different from the ones appearing above. On the other hand, it was proved in \cite{Junge-Groth} (see also \cite{Junge-Xu}) that $$\pi_1^o(id:OH_n\rightarrow OH_n)\simeq \sqrt{n(1+\ln n)},$$where $\simeq$ denotes equality up to universal constants and $\pi_1^o$ denotes the completely $1$-summing norm (see \cite{Pisier2}). Thus, defining the map $$\widehat{M}=j_n\circ j_n^*:M_N\rightarrow S_1^N$$it can be deduced that the associated tensor $M\in S_1^N\otimes S_1^N$ verifies $$\|M\|_{S_1^N\otimes_{\min} S_1^N}\simeq1 \text{   }\text{ and } \text{ }\|M\|_{S_1^N\widehat{\otimes} S_1^N}\simeq \sqrt{n(1+\ln n)}.$$Now, since the projection constant of $OH_n$ is of order $\sqrt{\frac{n}{1+ \ln n}}$ (see \cite{Junge-Groth}), we can find a map $P_n:S_1^N\rightarrow OH_n$ such that $\|P\|_{cb}\preceq \sqrt{\frac{n}{1+ \ln n}}$ and $P_\circ j_n=id_{OH_n}$. Then, following exactly the same argument as above one can show that the map $G$ also verifies Equation (\ref{lowerbound OH_n}) (with different constants). If we normalize, we obtain a rank-one quantum game $\tilde{M}:=\frac{M}{\sqrt{n(1+\ln n)}}$ with maximal value $V(\tilde{M})\simeq 1$, entangled value $\omega^*(\tilde{M}) \simeq \frac{1}{n(1 + \ln n)}$ and verifying Equation (\ref{lowerbound OH_n}).
\end{remark}

\

\

\

\


\begin{thebibliography}{99}
%
\bibitem{ALMSS} S. Arora, C. Lund, R. Motwani, M. Sudan and M. Szegedy, \emph{Proof verification and the hardness of approximation problems}. J. ACM, 45(3):501-555 (1998).

\bibitem{ArSa} S. Arora and S. Safra, \emph{Probabilistic checking of proofs: a new characterization of NP}, J. ACM, 45(1):70-122 (1998).

\bibitem{BAP} M. Ben-Or, A. Hassidim and H. Pilpel, \emph{Quantum Multiprover Interactive Proofs with Communicating Provers}, In Proc. 49th IEEE Symp. on Foundations of Computer Science (FOCS), pages 467-476 (2008).

\bibitem{BCFGGOS}  H. Buhrman, N. Chandran, S. Fehr, R. Gelles, V. Goyal, R. Ostrovsky and C. Schaffner,  \emph{Position-Based Quantum Cryptography: Impossibility and Constructions}, In Advances in Cryptology - CRYPTO 2011. Lecture Notes in Computer Science. Springer-Verlag, (2011).

\bibitem{chan} J.T. Chan, \emph{Facial Structure of the Trace Class}, Arch. Math. \textbf{64} (1994), no. 3, 185--187.

\bibitem{CSUU} R. Cleve, W. Slofstra, F. Unger and S. Upadhyay, \emph{Strong parallel repetition theorem for quantum XOR proof systems}.
In Proceedings of 22nd IEEE Conference on Computational Complexity, pages 282-299, (2007). quant-ph/0608146.

\bibitem{DJKR} I. Devetak, M. Junge, C. King and M. B. Ruskai, \emph{Multiplicativity of completely bounded p-norms implies a new additivity result},  Comm. Math. Phys. 266, no. 1, 37-63 (2006).

\bibitem{EJR} E. Effros, M. Junge and Z-J. Ruan, \emph{Integral mappings and the principle of local reflexivity for noncommutative
$L_1$-spaces}, Ann. of Math. 151, 59-92 (2000).

\bibitem{EffrosRuan} E. Effros and Z.-J. Ruan, Operator Spaces, London Mathematical Society Monographs New
Series, Vol. 23, Oxford University Press, Oxford, (2000).

\bibitem{Feige} U. Feige, \emph{On the success probability of two provers in one-round proof systems}, In Proc. 6th IEEE Structure in Complexity Theory, pages 116-123 (1991).

\bibitem{GLS} M. Grotschel, L. Lovasz and A. Schrijver, \emph{Geometric algorithms and combinatorial optimization,
volume 2 of Algorithms and Combinatorics} Springer-Verlag, Berlin, second edition, (1993).

\bibitem{Gutoski} G. Gutoski, \emph{Quantum Strategies and Local Operations},  arXiv:1003.0038 (2010).

\bibitem{HM1} U. Haagerup and M. Musat, \emph{On the best constants in noncommutative Khintchine-type inequalities}, J. Funct. Anal. 250, no. 2, 588-624 (2007).

\bibitem{HM} U. Haagerup and M. Musat, \emph{The Effros-Ruan conjecture for bilinear forms on C$^*$-algebras}, Invent. Math. 174, 139-163 (2008).

\bibitem{IKM} T. Ito, H. Kobayashi and K. Matsumoto.
\emph{Oracularization and two-prover one-round interactive proofs against nonlocal strategies}, Proceedings: 24th Annual IEEE Conference on Computational Complexity (CCC 2009), pp. 217-228 (2009).

\bibitem{JKP} N. Johnston, D. W. Kribs and V. I. Paulsen, \emph{Computing Stabilized Norms for Quantum Operations via the Theory of Completely Bounded Maps}, Quantum Information and Computation, 9, 16-35 (2009).


\bibitem{JKPP} N. Johnston, D.W. Kribs, V. I. Paulsen and R. Pereira, \emph{Minimal and Maximal Operator Spaces and Operator Systems in Entanglement Theory}, J. Funct. Anal. 260, 2407-2423 (2011).

\bibitem{Junge-Groth} M. Junge, \emph{Embedding of the operator space OH and the logarithmic ``little Grothendieck
inequality''}, Invent. Math. 161, no. 2, 225-286 (2005).


\bibitem{JNPPSW} M. Junge, M. Navascues, C. Palazuelos, D. P\'erez-Garc\'ia, V. B. Scholz and R. F. Werner, \emph{Connes' embedding problem and Tsirelson's problem}, J. Math. Phys 52, 012102 (2011).

\bibitem{JP} M. Junge and C. Palazuelos, \emph{Large violation of Bell inequalities with low entanglement}, Comm. Math. Phys. 306 (3), 695-746 (2011).

\bibitem{JPPVW2} M. Junge, C. Palazuelos, D. Pérez-García, I. Villanueva and M.M. Wolf, \emph{Operator Space theory: a natural framework for Bell inequalities}, Phys. Rev. Lett. 104, 170405 (2010).

\bibitem{Junge-Xu} M. Junge and Q. Xu, \emph{Representation of certain homogeneous Hilbertian operator spaces and applications}, Invent. Math. 179, no. 1, 75-118 (2010).

\bibitem{KKMTV} J. Kempe, H. Kobayashi, K. Matsumoto, B. Toner and T. Vidick, \emph{Entangled Games are Hard to approximate}, Proc. 49th FOCS'08, p. 447-456 (2008)

\bibitem{KeRe} J. Kempe and O. Regev, \emph{No Strong Parallel Repetition with Entangled and Non-signaling Provers}, In Proc. 25th IEEE Conference on Computational Complexity (CCC), pages 7-15 (2010).

\bibitem{KRT} J. Kempe, O. Regev and B. Toner, \emph{Unique Games with Entangled Provers are Easy}, In Proc. 49th IEEE Symp. on Foundations of Computer Science (FOCS), pages 457-466 (2008).

\bibitem{KeVi} J. Kempe and T. Vidick, \emph{Parallel Repetition of Entangled Games}, To appear in STOC'11.


\bibitem{Kitev} A. Kitaev, \emph{Quantum computations: Algorithms and error correction}. Russian Math. Surveys, 52(6):1191-1249 (1997).

\bibitem{KW} A. Kitaev and J. Watrous, \emph{Parallelization, amplification, and exponential time simulation of quantum interactive proof systems},
In Proceedings of the Thirty-second Annual ACM Symposium on the Theory of Computing, Portland OR, 608-617 (2000).

\bibitem{KoMa} H. Kobayashi and K. Matsumoto, \emph{Quantum multi-prover interactive proof systems with limited prior entanglement}. Journal of Computer and System Sciences 66(3), 429-450 (2003).


\bibitem{LTW} D. Leung, B. Toner and J. Watrous. \emph{Coherent state exchange in multi-prover quantum interactive proof systems}. arXiv:0804.4118[quant-ph].

\bibitem{Oikhberg} T. Oikhberg, \emph{Completely bounded and ideal norms of multiplication operators and Schur multipliers}, Integr. Equ. Oper. Theory 66, no. 3, 425-440 (2010).

\bibitem{Paulsen} V. Paulsen, \emph{Completely bounded maps and operator algebras}, Cambridge Studies in Advanced
Mathematics, vol. 78, Cambridge University Press, Cambridge, (2002).

\bibitem{PWJPV} D. Pérez-García, M. M. Wolf, C. Palazuelos, I. Villanueva and M. Junge, \emph{Unbounded violation of tripartite Bell inequalities}, Comm. Math. Phys. 279 (2), 455-486 (2008).

\bibitem{Pisier2} G. Pisier, \emph{Non-Commutative Vector Valued $L_p$-Spaces and Completely $p$-Summing Maps}, Asterisque, 247 (1998).

\bibitem{Pisier} G. Pisier, \emph{Introduction to operator space theory}, London Mathematical Society Lecture Note Series, 294. Cambridge University Press, Cambridge, 2003.

\bibitem{PisierSurvey} G. Pisier, \emph{Grothendieck's theorem, past and present}. To appear in Bull. Amer. Math. Soc. (N.S.), preprint available at http://arxiv.org/abs/1101.4195 (2011).

\bibitem{PS} G. Pisier and D. Shlyakhtenko, \emph{Grothendieck's theorem for operator spaces}, Invent. Math. 150 (2002), no. 1, 185-217.

\bibitem{RA} A. Rapaport and Ta-Shma, \emph{On the power of quantum, one round, two prover interactive proof systems},
Quantum Inf. Process. 6, no. 6, 445-459 (2007).

\bibitem{Raz} R. Raz, \emph{A parallel repetition theorem}, SIAM Journal on Computing, 27:763-803 (1998).

\bibitem{RegevVidick} O. Regev, T. Vidick, \emph{Quantum XOR games}, arXiv:1207.4939.

\bibitem{RegevVidickII} O. Regev, T. Vidick, \emph{Elementary Proofs of Grothendieck Theorems for Completely Bounded Norms}, Journal of Operator Theory, 2012. To appear. Also available at arXiv:1206.4025.

\bibitem{Simon} B. Simon, \emph{Pointwise domination of matrices and comparison of $S_p$ norms}, Pacific J. Math. 97 (1981), 471Ð475.

\bibitem{Smith} R.R. Smith, \emph{Completely bounded maps between C$^*$-algebras}, J. London Math. Soc. 27, 157-166 (1983).

\bibitem{Vidick} T. Vidick. \emph{Three-player entangled XOR games are NP-hard to approximate}, arXiv:1302.1242.

\bibitem{Wat} J. Watrous, \emph{Semidefinite programs for completely bounded norms}, Theory of computing, 5, 217-238 (2009).

\bibitem{Wat2} J. Watrous, \emph{PSPACE has constant-round quantum interactive proof systems}, Theoretical Computer Science 292(3), 575-588 (2003).
\end{thebibliography}
\end{document}